\documentclass[journal,onecolumn]{IEEEtran}

\usepackage{graphicx} 
\usepackage{latexsym}

\usepackage{amssymb,amsmath,amsthm} 
\usepackage{mathrsfs}
\usepackage{multicol}

\usepackage[T1]{fontenc}
\usepackage{amsfonts} \usepackage{bm} \usepackage{dsfont}
\usepackage{color} \usepackage[mathscr]{eucal} \usepackage{cite}
\usepackage{color} \usepackage{balance} \usepackage{url}
\usepackage{bbm} \usepackage{subfigure} \usepackage{subfigmat}

\usepackage{multirow} 

\usepackage{tabularx}

\usepackage{array} \newcolumntype{Z}{>{\centering\arraybackslash}
  m{.16\linewidth} }
\usepackage{array} \newcolumntype{W}{>{\centering\arraybackslash}
  m{.33\linewidth} }
\usepackage{array} \newcolumntype{A}{>{\centering\arraybackslash}
  m{.11\linewidth} }
\usepackage{array} \newcolumntype{V}{>{\centering\arraybackslash}
  m{.29\linewidth} }

\usepackage{array} \newcolumntype{B}{>{\arraybackslash}
  m{.16\linewidth} }

\usepackage{framed}

\newcommand{\eqdown}{\mathbin{\rotatebox[origin=c]{-90}{$=$}}}
\newcommand{\subsetdown}{\mathbin{\rotatebox[origin=c]{-90}{$\subset$}}}

\newcommand{\subsetneqdown}{\mathbin{\rotatebox[origin=c]{-90}{$\subsetneq$}}}

\theoremstyle{plain}
\newtheorem{theorem}{Theorem} 
\newtheorem{lemma}{Lemma}

\theoremstyle{definition}
\newtheorem{definition}{Definition}

\theoremstyle{remark}
\newtheorem{remark}{Remark}
\newtheorem{example}{Example}

\begin{document}

\title{The two-unicast problem}

\author{ \IEEEauthorblockN{Sudeep Kamath\IEEEauthorrefmark{1}, Venkat
    Anantharam\IEEEauthorrefmark{2}, David Tse\IEEEauthorrefmark{3},
    Chih-Chun Wang\IEEEauthorrefmark{4}\footnote{Parts of this paper
      were presented at the International Symposium on Network Coding
      2011, Beijing, China, the International Symposium on Information
      Theory 2013, Istanbul, Turkey, and the International Symposium
      on Information Theory 2014, Honolulu, Hawaii.}
  }\\
  \IEEEauthorblockA{\IEEEauthorrefmark{1}ECE Department, Princeton
    University,
    \\ sukamath@princeton.edu}\\
  \IEEEauthorblockA{\IEEEauthorrefmark{2}EECS Department, University
    of California, Berkeley, \\ ananth@eecs.berkeley.edu}\\
  \IEEEauthorblockA{\IEEEauthorrefmark{3}EE Department, Stanford
    University, \\ dntse@stanford.edu}\\
  \IEEEauthorblockA{\IEEEauthorrefmark{4}School of Electrical and
    Computer Engineering, Purdue University, \\ chihw@purdue.edu} }
\date{\today}
\maketitle
\thispagestyle{plain}
\pagestyle{plain}

\begin{abstract}
  We consider the communication capacity of wireline networks for a
  two-unicast traffic pattern. The network has two sources and two
  destinations with each source communicating a message to its own
  destination, subject to the capacity constraints on the directed
  edges of the network. We propose a simple outer bound for the
  problem that we call the Generalized Network Sharing (GNS) bound. We
  show this bound is the tightest edge-cut bound for two-unicast
  networks and is tight in several bottleneck cases, though it is not
  tight in general. We also show that the problem of computing the GNS
  bound is NP-complete. Finally, we show that despite its seeming
  simplicity, the two-unicast problem is as hard as the most general
  network coding problem. As a consequence, linear coding is
  insufficient to achieve capacity for general two-unicast networks,
  and non-Shannon inequalities are necessary for characterizing
  capacity of general two-unicast networks.
\end{abstract}

\section{Introduction}
\label{sec:intro}


The holy grail of network information theory is the characterization
of the information capacity of a general network. While there has been
some success towards this goal, a complete capacity characterization
has been open for even simple networks such as the broadcast channel,
the relay channel, and the interference channel. 

The seminal work of Ahlswede et al. \cite{ACLY00} characterized the
information capacity of a family of networks assuming a simple
\emph{network model} ((i.e. assuming a directed wireline network where
links between nodes are unidirectional, orthogonal and noise-free),
and a simple \emph{traffic pattern} (i.e. multicast, where the same
information is to be transmitted from one source node to several
destination nodes). Under this network model, the complex aspects of
real-world communication channels, such as broadcast, superposition,
interference, and noise are absent. Similarly, the traffic pattern is
simple and ensures that there is no interference from multiple
messages. \cite{ACLY00} showed that for multicast in directed wireline
networks, a simple outer bound on capacity, namely the cutset bound
\cite{ElGamal81} was achievable, completing the capacity
characterization. Later, \cite{LYC03} and \cite{KoetterMedard} showed that a
simple class of coding strategies - linear network coding over a
finite field - can achieve the multicast capacity of a wireline
network. In spite of the simplicity of this network model,
understanding the capacity region of multicast in directed wireline
networks has proved very useful in offering insight into capacity and
coding strategies for other network models, such as Gaussian networks
\cite{KannanRV11}. Subsequently, the directed wireline network
capacity characterization problem was solved for the traffic patterns
of two-level multicast \cite{KoetterMedard} (i.e. a source node
produces $k$ messages, each message to be delivered at exactly one
destination (first-level), and in addition, a collection of nodes
(second-level) requiring \emph{all} the messages) and two-receiver
multicast with private and common data \cite{ErezFeder03,
  NgaiYeung04}. In both these cases, the cutset bound was shown to be
tight.

However, it was soon discovered that the capacity characterization of
a \emph{general traffic pattern} for the \emph{simplified network
  model} of directed wireline networks was still a problem of
considerable difficulty. \cite{Zeger_insufficiency} showed that linear
codes were not sufficient to achieve capacity for the general traffic
pattern. \cite{Zeger_matroid} showed that the so-called LP bound
\cite{Yeung99}, which is a computable outer bound on the capacity,
tighter than the cutset bound, derived from all possible so-called
Shannon-type inequalities, was not tight in general. \cite{ChanGrant}
showed that if we can compute the network capacity region for the
general traffic pattern, then we can characterize all the so-called
non-Shannon information inequalities. However, the networks presented
as \emph{counterexamples} in all the above works either have more than
two sources or more than two destinations, often much more. It is a
natural question to ask whether the difficulty of the problem stems
from this. This is hardly an unusual sentiment since many different
but related problems enjoy a simplicity with 2 users that is not
shared by the corresponding problems with 3 or more users. For
instance:
\begin{itemize}
\item For two-unicast \emph{undirected} networks, the cutset bound is
  tight but this is not the case for three-unicast undirected networks
  \cite{Hu63}.
\item The capacity of two-user interference channels is known to
  within one bit \cite{ETW08} but no such result is known for
  three-user interference channels.
\item The capacity region up to unit rates for layered linear
  deterministic networks has been characterized for two-unicast
  networks \cite{WangKamathTse11} but not for three-unicast networks.
\item The degrees of freedom for layered wireless networks is known
  for two-unicast networks \cite{ShomoronyAvestimehr13} but not for
  three-unicast networks.
\item There are no known analogs with three-receivers for the
  aforementioned two-receiver multicast problem with private and
  common data \cite{ErezFeder03, NgaiYeung04}.
\end{itemize}

The central candidate for the simplest unsolved problem in capacity of
wireline networks is the two-unicast traffic pattern, i.e. the problem
of communication between two sources and two destinations, each source
with an independent message for its own destination. The only complete
capacity result in the literature dealing with the two-unicast network
capacity is \cite{Wang10} which characterizes the necessary and
sufficient condition for achieving $(1,1)$ in a two-unicast network
with all links having integer capacities. This result unfortunately
relies heavily on the assumption of integer link capacities, and hence
cannot give us necessary and sufficient conditions for achieving other
points such as $(2,2)$ or $(3,3)$ by scaling of link capacities. This
success with the $(1,1)$ rate pair in two-unicast networks stands in
strong contrast with the intractability of the general $k$-unicast
problem \cite{Zeger_insufficiency, Zeger_matroid, ChanGrant}.
Although the two-unicast capacity characterization for general rates
$(R_1,R_2)$ remains open, it is often believed that the two-unicast
problem enjoys a similar simplicity as other two-user information
theoretic problems. There are many existing results that aim to
characterize the general achievable rate region for the two-unicast
problem (not limited to the (1,1) case in \cite{Wang10}) and/or the
$k$-unicast problem with small $k.$ For example,
\cite{HuangRamamoorthy11}, \cite{HuangRamamoorthy12}, and
\cite{HuangRamamoorthy11b} study capacity of two-unicast,
three-unicast, and $k$-unicast networks respectively, from a
source-destination cut-based analysis. The authors of \cite{ZCM12}
present an edge-reduction lemma using which they compute an achievable
region for two-unicast networks. In a subsequent work \cite{ZCM13}
they show that the Generalized Network Sharing bound that we will
study in this paper gives necessary and sufficient conditions for
achievability of the $(N,1)$ rate pair in a special class of
two-unicast networks (networks with Z-connectivity and satisfying
certain richness conditions). Unfortunately, none of the above results
is able to fully characterize the capacity region for general
two-unicast networks even for the second simplest instance of the rate
pair $(1,2)$, let alone the capacity region for three-unicast or
$k$-unicast networks (for small $k$). Such a relatively frustrating
lack of progress prompts us to re-examine the problem at hand and
investigate whether the lack of new findings is due to the inherent
hardness of the two-unicast problem.

In this paper, we have contributions along two main themes. Along the
first theme, we present and investigate a new outer bound for the
two-unicast problem that is stronger than the cutset bound. This new
bound is a simple improvement over the Network Sharing outer bound of
\cite{YYZ06}, and we call it the Generalized Network Sharing (GNS)
outer bound. We observe that the GNS bound is the tightest edge-cut
bound for the two-unicast problem, and is tight in various
``bottleneck'' cases. However, we find that the GNS bound is
NP-complete to compute and is also not tight for the two-unicast
problem. Along the second theme, we show that the lack of success so
far with a complete characterization of capacity for the two-unicast
problem is due to its inherent hardness. We show that the two-unicast
problem with a general rate pair $(R_1,R_2)$ captures all the
complexity of the multiple unicast problem, i.e. solving the
two-unicast problem for general rate pairs is as hard as solving the
$k$-unicast problem for any $k\geq 3.$ Thus, the two-unicast problem
is the ``hardest network coding problem''. We show that given any
multiple-unicast rate tuple, there is a rate tuple with suitable
higher rates but \emph{fewer} sources that is more difficult to
ascertain achievability of. In particular, we show that solving the
well-studied but notoriously hard $k$-unicast problem with unit-rates
(eg. \cite{Zeger_insufficiency}, \cite{Zeger_matroid}) is no harder
than solving the two-unicast problem with rates $(k-1,k).$
Furthermore, by coupling our results with those of
\cite{Zeger_insufficiency}, \cite{Zeger_matroid}, we show the
existence of a two-unicast network for which linear codes are
insufficient to achieve capacity, and a two-unicast network for which
non-Shannon inequalities can provide tighter outer bounds on capacity
than Shannon-type inequalities alone.

We mention here that since it first appeared in \cite{KTA11}, the GNS
bound has found three distinct interpretations: an algebraic
interpretation \cite{ZCM12}, a network concatenation interpretation
\cite{ShomoronyAvestimehr14}, and a maximum acyclic subgraph bound
through a connection with index coding \cite{ShanmugamDimakis14}. The
bound has also found a number of applications,
eg. \cite{KamathKannanViswanath14}, \cite{ZCM13}, \cite{WangChen14}.

The rest of the paper is organized as follows. In
Sec.~\ref{sec:prelim}, we set up preliminaries and notation. We
propose an improved outer bound for the two-unicast problem that we
call the Generalized Network Sharing (GNS) bound, and study its properties
in Sec.~\ref{sec:GNS-bound}. We show that the two-unicast problem with
general rate pairs is as hard as the $k$-unicast problem with
$k\geq 3,$ in Sec.~\ref{sec:two-unicast-hard}.

\section{Preliminaries}
\label{sec:prelim}

A $k$-unicast network $\mathcal{N}$ consists of a directed acyclic
graph $\mathcal{G}=(\mathcal{V},\mathcal{E})$ ($\mathcal{V}$ being the
vertex set and $\mathcal{E}$ being the edge set), along with an
assignment of edge-capacities
$\underbar{C}=(C_e)_{e\in\mathcal{E}(\mathcal{G})}$ with
$C_e\in\mathbb{R}_{\geq 0} \cup \{\infty\} \ \forall
e\in\mathcal{E}(\mathcal{G}).$
It has $k$ distinguished vertices $s_1,s_2,\ldots,s_k$ called sources
(not necessarily distinct) and $k$ distinguished vertices
$t_1,t_2,\ldots, t_k$ called destinations (not necessarily distinct
from each other or from the sources), where source $s_i$ has
independent information to be communicated to destination
$t_i, i=1,2,\ldots,k.$

For edge $e=(v,v^\prime)\in\mathcal{E}(\mathcal{G})$, define
$\mathsf{tail}(e):=v$ and $\mathsf{head}(e):=v^\prime,$ the edge being
directed from the tail to the head.  For
$v\in\mathcal{V}(\mathcal{G}),$ let $\mathsf{In}(v)$ and
$\mathsf{Out}(v)$ denote the edges entering into and leaving $v$
respectively.

For $S\subseteq\mathcal{E}(\mathcal{G}),$ define $C(S):=\sum_{e\in S} C_e.$
For sets $A,B\subseteq\mathcal{V}(\mathcal{G}),$ we say
$S\subseteq\mathcal{E}(\mathcal{G})$ is an $A-B$ cut if there is no directed
path from any vertex in $A$ to any vertex in $B$ in the graph
$\mathcal{G}\setminus S.$ Define the \emph{mincut} from $A$ to $B$ by
$c(A;B):=\min\left\{C(S): S\mbox{\ is an }A-B\mbox{ cut}\right\},$
where by convention, $c(A;B) = 0,$ if there are no directed paths from
$A$ to $B$ or if either $A$ or $B$ is empty, and also $c(A;B) =
\infty$ if $A\cap B\neq \emptyset.$

\begin{definition}\label{def:achievable}
Given a $k$-unicast network $\mathcal{N}=(\mathcal{G},\underbar{C})$
for source-destination pairs $\{(s_i;t_i)\}_{i=1}^k,$ we say
that the non-negative rate tuple $(R_1,R_2,\ldots,R_k)$ is
\textit{achievable}, if for any $\epsilon>0,$ there exist positive
integers $N$ and $T$ (called block length and number of epochs
respectively), a finite alphabet $\mathcal{A}$ with $|\mathcal{A}|\geq
2$ and using notation $H_v:=\Pi_{i: v=s_i}
\mathcal{A}^{\lceil NTR_i\rceil}$ (with an empty product being the
singleton set),
  \begin{itemize}
  \item encoding functions for $1\leq t\leq T, e=(u,v)\in\mathcal{E},$ \\
    $f_{e,t}:H_u\times \Pi_{e^\prime\in\mathsf{In}(u)}\left(\mathcal{A}^{\lfloor NC_{e^\prime}\rfloor}\right)^{(t-1)}\mapsto\mathcal{A}^{\lfloor NC_e\rfloor},$
  \item decoding functions at destinations $t_i$ for $i\in\mathcal{I},$ \\
    $f_{t_i}:H_{t_i}\times \Pi_{e^\prime\in\mathsf{In}(t_i)}\left(\mathcal{A}^{\lfloor NC_{e^\prime}\rfloor}\right)^T\mapsto\mathcal{A}^{\lceil NTR_i\rceil}$
  \end{itemize}

  \noindent  with the property that under the uniform probability distribution on  $\Pi_{i\in\mathcal{I}} \mathcal{A}^{\lceil NTR_i\rceil},$
\begin{align}
&\mathsf{Pr}\left(g(m_1, m_2, \ldots,m_k)\neq (m_1, m_2, \ldots,m_k)\right)\leq \epsilon,
\end{align}
\noindent where $g:\Pi_{i\in\mathcal{I}} \mathcal{A}^{\lceil
  NTR_i\rceil}\mapsto \Pi_{i\in\mathcal{I}} \mathcal{A}^{\lceil
  NTR_i\rceil}$ is the global decoding function induced inductively by
$\{f_{e,t}:e\in\mathcal{E}(\mathcal{G}), 1\leq t\leq T\}$ and
$\{f_{t_i}: i=1,2,\ldots,k\}.$ 
\end{definition}

The \emph{Shannon capacity} (also
simply called \emph{capacity}) of a $k$-unicast network $\mathcal{N},$
denoted
$\mathcal{C}(\mathcal{N})=\mathcal{C}(\mathcal{G},\underbar{C}),$ is
defined as the closure of the set of achievable rate tuples. The
closure of the set of achievable rate tuples over choice of
$\mathcal{A}$ as any finite field and all functions being linear
operations on vector spaces over the finite field, is called the
\emph{vector linear coding Shannon capacity} or simply \emph{linear
  coding capacity}. If we further have $N=1,$ then the convex closure
of achievable rate tuples is called the \emph{scalar linear coding
  capacity}.

The \emph{zero-error capacity}, \emph{zero-error linear coding
  capacity} are similarly defined as the closure of the set of
achievable rate tuples with zero error using general and vector linear
codes respectively. Finally, the \emph{zero-error exactly achievable
  region} and \emph{zero-error exactly achievable linear coding
  region} are similarly defined as the set of achievable rate tuples
with zero error (no closure taken), using general and vector linear
codes respectively. This notion of exact achievability when studied
with unit rates, is also called \emph{solvability} in the literature
\cite{Zeger_nonreversibility}.

Table~\ref{tb:capacity} summarizes these definitions.
{\large
  \begin{table*}[ht] 
    {\begin{center}
      \begin{tabularx}{\textwidth}{Z|WAV} 
        & & \\
        & \begin{center}Vector Linear codes\end{center} & & \begin{center}General codes\end{center}
        \\
        & & \\
        \hline 
        & & \\ 
        \begin{center}Zero-error Exact Achievability\end{center} 

      & \begin{center}\begin{framed}Set of linearly achievable rate
        tuples with zero error\end{framed}\end{center} 

        & \begin{center}{\LARGE $\subsetneq$}

          ($\neq$ from \cite{Zeger_insufficiency})
          \end{center}
                
        & \begin{center}\begin{framed}Set of achievable rate
            tuples with zero error\end{framed} \end{center} 

        \\
        &  {\LARGE $\subsetdown$} 

        ($=, \neq$ unknown)
        & & {\LARGE $\subsetneqdown$} 

        ($\neq$ from \cite{Zeger_unachievability}) \\
        & & \\
        \begin{center} Zero-error Capacity \end{center} 

        & \begin{center}\begin{framed}Closure of set of linearly
            achievable rate tuples with zero
            error\end{framed}\end{center}

    &\begin{center}{\LARGE $\subsetneq$}

          ($\neq$ from \cite{Zeger_insufficiency})
        \end{center}


        & \begin{center}\begin{framed}Closure of set of achievable rate
        tuples with zero error\end{framed}\end{center} 
      
      \\
        &  {\LARGE $\eqdown$} 

        ($=$ from simple argument; see Remark~\ref{rem:linear-codes}
        at end of Appendix~\ref{subsec:two-unicast-hard})
        & & {\LARGE $\subsetdown$} 

        ($=,\neq$ unknown \cite{LangbergEffros11}) \\
      & & \\
        \begin{center}Shannon Capacity\end{center} 

      & \begin{center}\begin{framed}Closure of set of linearly
        achievable rate tuples with vanishing error\end{framed}\end{center}

    &\begin{center}{\LARGE $\subsetneq$}

          ($\neq$ from \cite{Zeger_insufficiency})
        \end{center}
    
       & \begin{center}\begin{framed}Closure of set of achievable rate
        tuples with vanishing error\end{framed}\end{center} 

        \\
        & & \\

      \end{tabularx}
    \end{center}
  }
   \caption{Different notions of capacity}
    \label{tb:capacity}
  \end{table*}
}


\section{Generalized Network Sharing Outer Bound}
\label{sec:GNS-bound}

We introduce the Generalized Network Sharing (GNS) bound as a simple
outer bound on the Shannon capacity of a two-unicast network that is
tighter than the cutset bound \cite{ElGamal81}.

We say a set of edges $S\subseteq\mathcal{E}(\mathcal{G})$ form a
\textit{Generalized Network Sharing cut (GNS-cut)} if
\begin{itemize}
\item $\mathcal{G}\setminus S$ has no paths from $s_1$ to $t_1,$ $s_2$ to
  $t_2$ and $s_2$ to $t_1$ OR
\item $\mathcal{G}\setminus S$ has no paths from $s_1$ to $t_1,$ $s_2$ to
  $t_2$ and $s_1$ to $t_2.$
\end{itemize}

We adopt the natural convention that if $u$ and $v$ are the same node,
then, no cut separates $u$ from $v.$ Thus, if $s_2 = t_1,$ then a set
of edges $S$ forms a GNS-cut only if it removes all paths from $s_1$
to $t_1,$ $s_2$ to $t_2$ and $s_1$ to $t_2.$


\begin{theorem}\label{thm:gns-bound} (GNS outer bound)
  For a two-unicast network $\mathcal{N}=(\mathcal{G},\underbar{C})$ and a GNS
  cut $S\subseteq\mathcal{E}(\mathcal{G}),$ we have $R_1+R_2\leq C(S)\
  \forall\ (R_1,R_2)\in\mathcal{C}(\mathcal{N}).$
\end{theorem}

\begin{remark}
  Our name for the Generalized Network Sharing bound is derived from
  an earlier bound in the literature called the Network Sharing bound
  \cite{YYZ06} which may be described as follows: Fix $(i,j)=(1,2)$ or
  $(2,1).$ For a two-unicast network
  $\mathcal{N}=(\mathcal{G},\underbar{C}),$ if
  $T\subseteq\mathcal{E}(\mathcal{G})$ is an $\{s_1,s_2\}-\{t_1,t_2\}$
  cut and if $S\subseteq T$ is such that for each edge
  $e\in T\setminus S,$ we have that $\mathsf{tail}(e)$ is reachable
  from $s_i$ but not from $s_j$ in $\mathcal{G}$(*) and
  $\mathsf{head}(e)$ can reach $t_j$ but not $t_i$ in $\mathcal{G},$
  then we have
  $R_1+R_2\leq C(S)\ \forall (R_1,R_2)\in\mathcal{C}(\mathcal{N}).$ If
  the (*) were replaced by $\mathcal{G}\setminus S,$ then we would get
  the GNS bound. Thus, the improvement in the bound is very subtle but
  important.

  The proof of the GNS bound relies on the same idea that was used to
  prove the Network Sharing bound. The GNS bound subsumes the Network
  Sharing bound and it can be strictly tighter, as shown by the grail
  network in Fig.~\ref{fig:GNS_better_than_NS}. In
  Theorem~\ref{thm:gns-tightest}, we will show that the GNS bound is
  fundamental to the two-unicast problem: in the class of edge-cut
  bounds, it is the tightest and cannot be further improved upon.

  \begin{figure}[htbp]
    \begin{center}
      \subfigure[\emph{Grail} network with general edge
      capacities]{\includegraphics[width=1.5in,
        height=!]{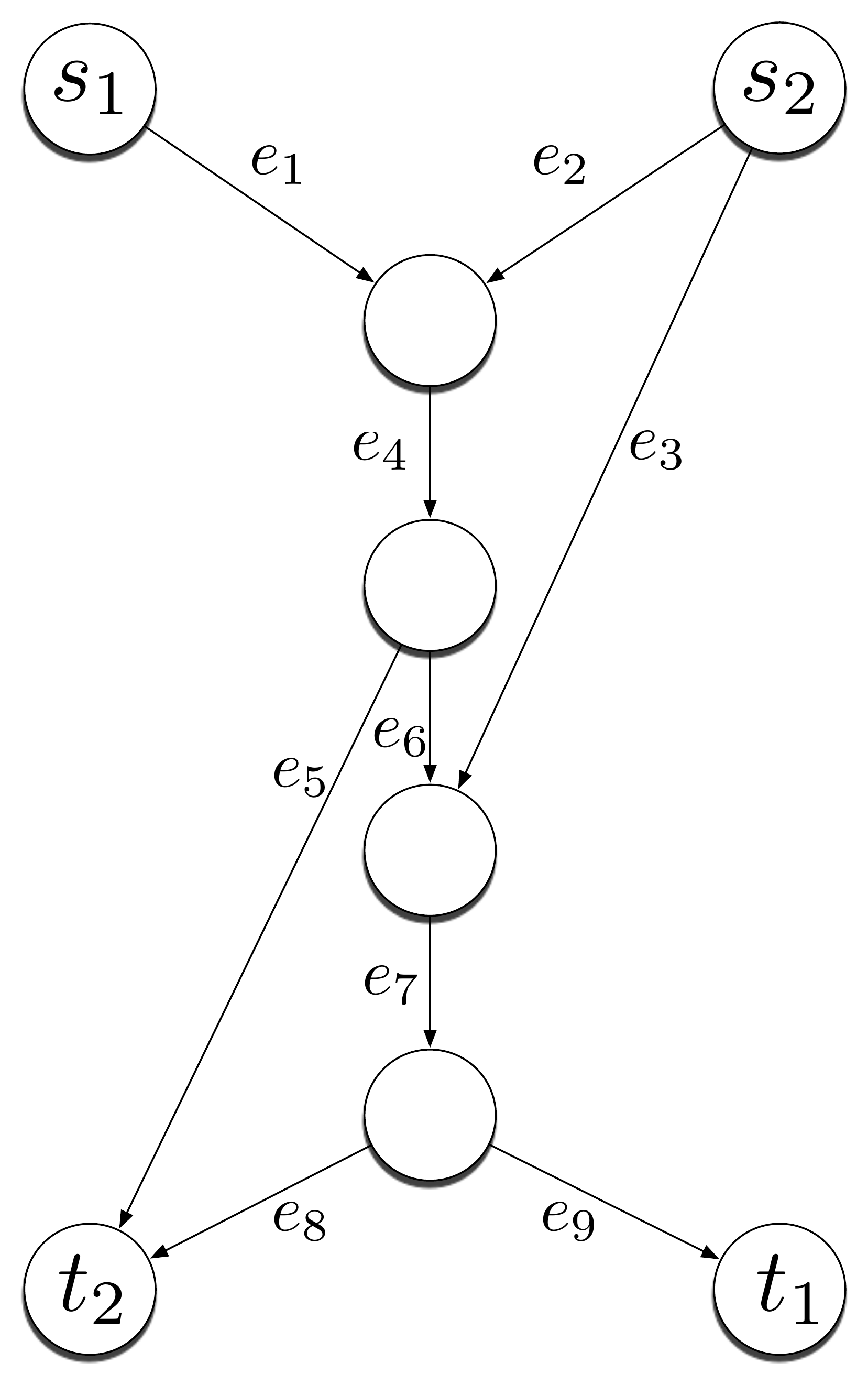}}\hspace{25pt}
      \subfigure[Grail with variable edge capacities. $e_2, e_7$ have
      unit capacity and all other edges have capacity 2
      units.]{\includegraphics[width=1.5in,
        height=!]{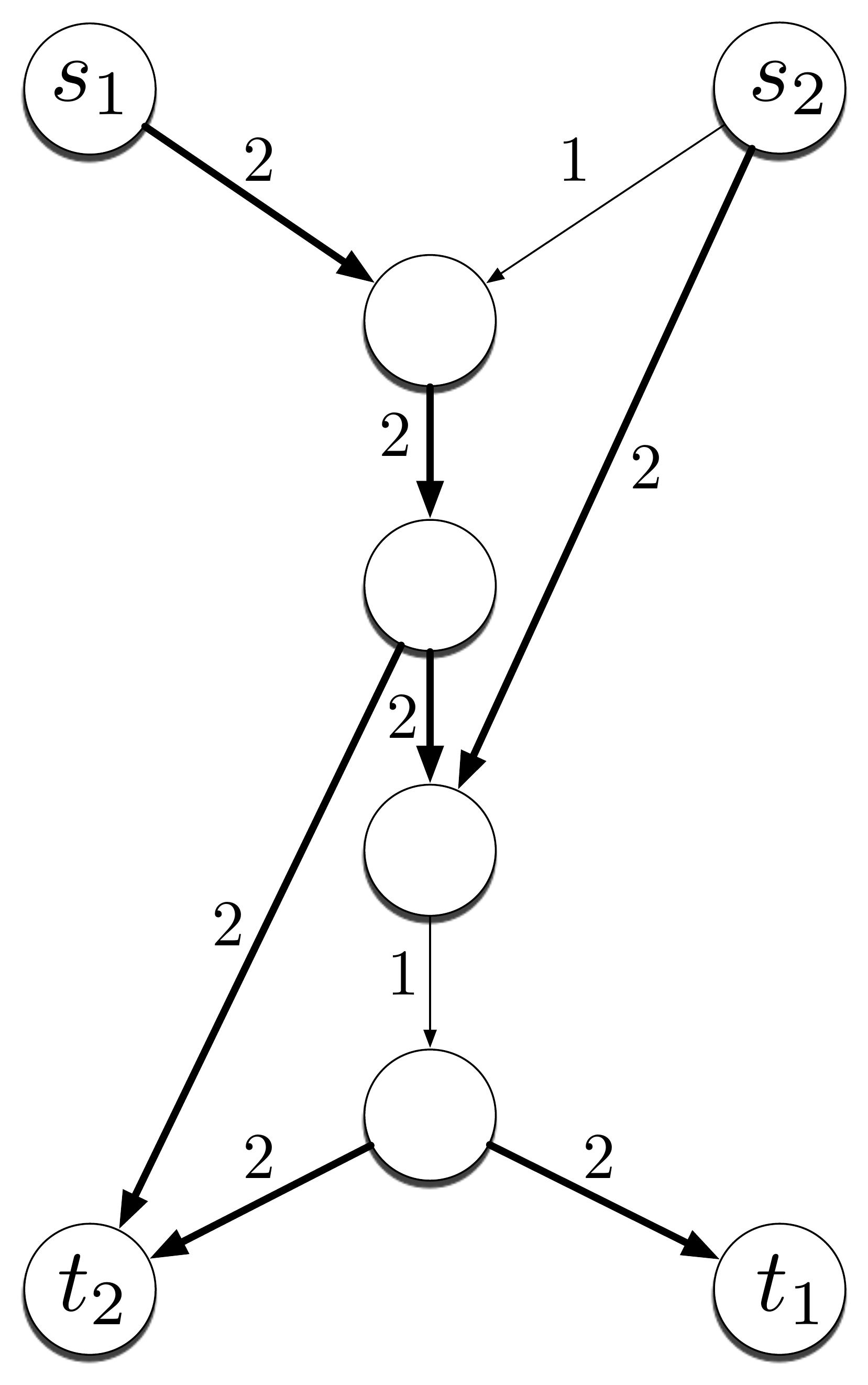}}\hspace{25pt}
      \subfigure[Network Sharing outer bound and GNS outer bound for
      the network in (b)]{\includegraphics[width=2.5in,
        height=!]{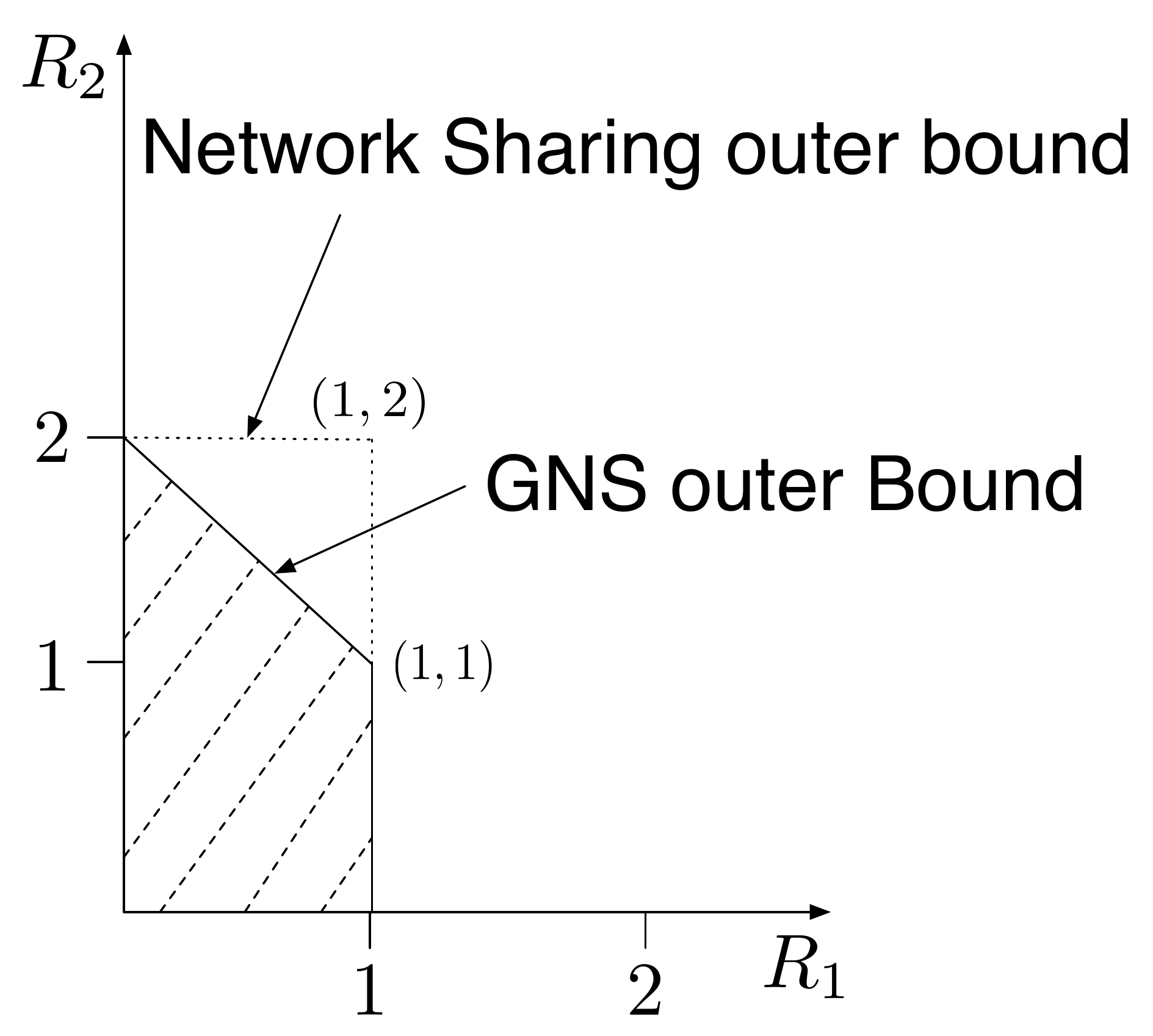}}
      \caption{The GNS outer bound can be strictly better than the
        Network Sharing outer bound \cite{YYZ06}. $\{e_2,e_7\}$ is a
        GNS-cut.}
      \label{fig:GNS_better_than_NS}
    \end{center}
  \end{figure}

\end{remark}
\begin{proof}
  Suppose for a set of edges $S\subseteq\mathcal{E}(\mathcal{G}),$
  $\mathcal{G}\setminus S$ has no paths from $s_1,s_2$ to $t_1$ and no
  paths from $s_2$ to $t_2.$ Fixing $0<\epsilon<\frac 12,$ consider a
  scheme of block length $N,$ achieving the rate pair $(R_1,R_2)$ over
  alphabet $\mathcal{A}$ with error probability at most
  $\epsilon$. Let $W_1,W_2$ be independent and distributed uniformly
  over the sets $\mathcal{A}^{\lceil NR_1\rceil}$ and
  $\mathcal{A}^{\lceil NR_2\rceil}$ respectively. For each edge $e,$
  define $X_e$ as the concatenated evaluation of the functions
  specified by the scheme for edge $e.$ For
  $S\subseteq\mathcal{E}(\mathcal{G}),$ let $X_S:=(X_e)_{e\in S}.$ To
  simplify calculations, we will assume all logarithms are to base
  $|\mathcal{A}|.$

  As $\mathcal{G}\setminus S$ has no paths from $s_1$ or $s_2$ to
  $t_1,$ it follows that the message $W_1$ can be recovered
  successfully from $X_S$ with probability at least $1-\epsilon.$
  Then, by Fano's inequality,
  \begin{align}
    H(W_1|X_S) & \leq h(\epsilon)+\epsilon\lceil NR_1\rceil.
  \end{align}
  Similarly, since $\mathcal{G}\setminus S$ has no paths from $s_2$ to
  $t_2,$ it follows that the message $W_2$ can be recovered
  successfully from $X_S$ and $W_1$ with probability at least
  $1-\epsilon.$ Fano's inequality gives
  \begin{align}
    H(W_2|W_1,X_S) & \leq h(\epsilon)+\epsilon\lceil NR_2\rceil.
  \end{align}
  This gives
  \begin{align}
    N(R_1+R_2)& \leq H(W_1,W_2) \\
    & = I(W_1,W_2;X_S) + H(W_1,W_2|X_S) \\
    & = I(W_1,W_2;X_S) + H(W_1|X_S) + H(W_2|W_1,X_S) \\
    & \leq H(X_S) + 2h(\epsilon) + \epsilon (N(R_1+R_2)+2) \\
    & \leq NC(S) + 2(\epsilon+h(\epsilon)) +
    N\epsilon(R_1+R_2) \\
    R_1+R_2 & \leq C(S) + \frac{2(\epsilon+h(\epsilon))}{N} +
    \epsilon(R_1+R_2) \\
    & \leq C(S) + 2(\epsilon+h(\epsilon)) +
    \epsilon(R_1+R_2).\label{eq:last}
  \end{align}
  Since $\epsilon$ can be made arbitrarily small by a suitable coding
  scheme with a suitable block length, from \eqref{eq:last}, we must
  have $R_1+R_2\leq C(S).$ As this inequality holds for every
  vanishing error achievable rate pair $(R_1,R_2),$ it also holds for
  every point in the closure of the set of vanishing error achievable
  rate pairs.
\end{proof}

For a given two-unicast network
$\mathcal{N}=(\mathcal{G},\underbar{C}),$ let the \textit{GNS sum-rate
  bound} $c_\mathsf{gns}(s_1,s_2;t_1,t_2)$ be defined as
$c_\mathsf{gns}(s_1,s_2;t_1,t_2)
:=\min\{C(S):S\subseteq\mathcal{E}(\mathcal{G})\mbox{ is a
  GNS-cut}\}.$ The \textit{GNS outer bound} is defined as the region
$\{(R_1,R_2): R_1\leq c(s_1;t_1), R_2\leq c(s_2;t_2), R_1+R_2\leq
c_\mathsf{gns}(s_1,s_2;t_1,t_2)\}.$ Note that the GNS sum-rate bound
is a number while the GNS outer bound is a region.

Before moving to properties of the GNS bound, we briefly remark that
the GNS bound may be extended to multiple unicast networks as stated
in Theorem~\ref{thm:gns-bound-multiple-unicast} below. The proof is
very similar to that of Theorem~\ref{thm:gns-bound} and is omitted.
\begin{theorem}\label{thm:gns-bound-multiple-unicast}
  Consider a $k$-unicast network $\mathcal{N}=(\mathcal{G},\underbar{C}).$
  For non-empty $I\subseteq \{1,2,\ldots,k\}$ and
  $S\subseteq\mathcal{E}(\mathcal{G}),$ suppose there exists a bijection
  $\pi:I\mapsto \{1,2,\ldots,|I|\}$ such that $\forall\ i,j\in I,$
  $\mathcal{G}\setminus S$ has no paths from source $s_i$ to destination
  $t_j$ whenever $\pi(i)\geq \pi(j). $ Then,
$$\sum_{i\in I} R_i\leq C(S)\ \forall (R_1,R_2,\ldots,R_n)\in\mathcal{C}(\mathcal{N}).$$
\end{theorem}

\subsection{GNS bound is the tightest edge-cut outer bound for
  two-unicast}\label{subsec:GNS-tightest-edge-cut}
For an uncapacitated two-unicast network $\mathcal{G}$ (i.e. a
two-unicast graph), an inequality of the form $\alpha_1 R_1 + \alpha_2
R_2 \leq C(S),$ with $\alpha_1, \alpha_2\in\{0,1\},
S\subseteq\mathcal{E}(\mathcal{G})$ is called an \emph{edge-cut bound}
if the inequality holds for all
$(R_1,R_2)\in\mathcal{C}(\mathcal{G},\underbar{C}),$ for each choice
of edge capacities $\underbar{C}.$ The cutset outer bound
\cite{ElGamal81}, the Network Sharing bound \cite{YYZ06} and the GNS
bound are all collections of edge-cut bounds.

\begin{theorem} \label{thm:gns-tightest} Let $\mathcal{G}$ be an
  uncapacitated two-unicast network, and let
  $S\subseteq\mathcal{E}(\mathcal{G})$ be such that $R_1+R_2\leq C(S)$
  is an edge-cut bound, i.e. $R_1+R_2\leq C(S)$ holds for all
  $(R_1,R_2)\in\mathcal{C}(\mathcal{G},\underbar{C})$ for all choices
  of $\underbar{C}.$ Then, exactly one of the following is true:
  \begin{itemize}
  \item $S$ is a GNS-cut
  \item $S$ is not a GNS-cut but $c(s_1;t_1)+c(s_2;t_2)\leq C(S)$ for
    all choices of $\underbar{C}.$
  \end{itemize}
\end{theorem}

\begin{remark}
  Since the cutset bound is tight for single unicast, the cutset
  bounds provide all possible edge-cut bounds on the individual
  rates. Theorem~\ref{thm:gns-tightest} says that the GNS cuts
  together provide all possible edge-cut bounds on the sum rate that
  are not already implied by the individual rate cutset bounds.
\end{remark}

\begin{proof} Suppose $R_1+R_2\leq C(S)$ holds for all
  $(R_1,R_2)\in\mathcal{C}(\mathcal{G},\underbar{C})$ for all choices
  of $\underbar{C}.$ Then, clearly $\mathcal{G}\setminus S$ has no
  paths from $s_1$ to $t_1$ and no paths from $s_2$ to $t_2.$ Suppose
  $S$ is not a GNS-cut so that $\mathcal{G}\setminus S$ has no paths
  from $s_1$ to $t_1$ or from $s_2$ to $t_2$ but it has paths from
  $s_1$ to $t_2$ and $s_2$ to $t_1.$ Define
  $C_i(S):=\min\{C(T): T\subseteq S,\ T \mbox{ is an }s_i-t_i \mbox{
    cut} \}$
  for $i=1,2.$ Fix any choice of non-negative reals $\{c_e:e\in S\}.$
  Consider the following choice of link capacities:
  $C_e= c_e\ \forall e\in S$ and $C_e= \infty\ \forall e\notin S.$ For
  this choice of link capacities, the individual rate mincuts are
  given by $c(s_i;t_i)=C_i(S), i=1,2.$ We will use the following lemma
  proved in Appendix~\ref{subsec:lem_two_multicast}.

  \begin{lemma}\label{lem:two-multicast} (Two-Multicast Lemma)
    For a two-multicast network
    $\mathcal{N}=(\mathcal{G},\underbar{C})$ with sources $s_1$ and
    $s_2$ multicasting independent messages at rates $R_1$ and $R_2$
    respectively to be recovered at both the destinations $t_1$ and
    $t_2,$ the capacity region is given by
    \begin{align*}
      R_1 & \leq \min\{c(s_1;t_1), c(s_1;t_2)\}, \\
      R_2 & \leq \min\{c(s_2;t_1), c(s_2;t_2)\},\\
      R_1 + R_2 & \leq \min\{c(s_1,s_2;t_1), c(s_1,s_2;t_2)\}.
    \end{align*}
  \end{lemma}
  \vspace{10pt}
  By Lemma~\ref{lem:two-multicast},
  $(R_1,R_2)$ is achievable for two-multicast from $s_1,s_2$ to
  $t_1,t_2$ if and only if $R_1\leq C_1(S)$ and $R_2\leq C_2(S),$
  since $c(s_1,s_2;t_1)\geq c(s_2;t_1) = \infty, c(s_1,s_2;t_2)\geq
  c(s_1;t_2) = \infty.$ Thus, $(C_1(S), C_2(S))$ is achievable for
  two-multicast and hence, also for two-unicast. Since $R_1+R_2\leq
  C(S)$ holds for all $(R_1,R_2)\in\mathcal{C}(\mathcal{G},\underbar{C})$, we
  must have $C_1(S) + C_2(S)\leq C(S)\ \forall \{C_e:e\in S\}.$
  This is a purely graph theoretic property about the structure of the
  set of edges $S$ relative to the uncapacitated network $\mathcal{G}.$
  Now, for an arbitrary assignment of link capacities $\underbar{C},$
  we have by definition, $c(s_1;t_1)\leq C_1(S)$ and $c(s_2;t_2)\leq
  C_2(S).$ Thus, we have $c(s_1;t_1)+c(s_2;t_2)\leq C(S).$
\end{proof}

\begin{example}
  Consider the butterfly network in
  Fig.~\ref{fig:butterfly_network}. One can show that $R_1+R_2\leq
  C_{e_1}+C_{e_2}$ is an edge-cut bound, that is it holds for each
  $(R_1,R_2)\in\mathcal{C}(\mathcal{G},\underbar{C}),$ for each choice
  of edge-capacities $\underbar{C}.$ However, $\{e_1,e_2\}$ is not a
  GNS-cut. So, in accordance with Theorem~\ref{thm:gns-tightest}, this
  edge-cut bound must be implied by the individual rate cutset bounds
  and indeed it follows from the cutset bounds $R_1\leq C_{e_1},
  R_2\leq C_{e_2}.$

  \begin{figure}[h]
    \begin{center}
      \includegraphics[width = 1.5in,
      height=!]{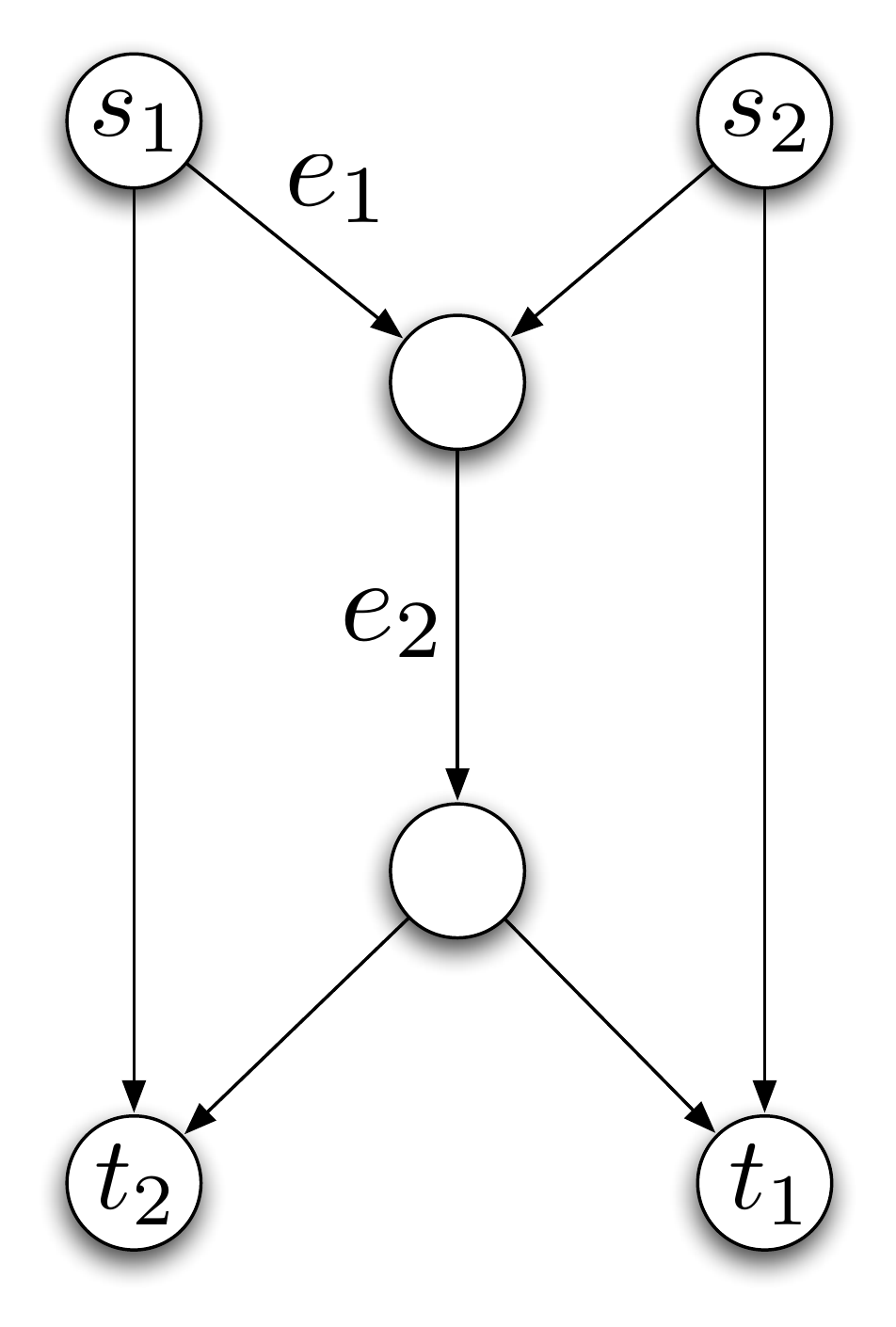}
      \caption{Butterfly Network}
      \label{fig:butterfly_network}
    \end{center}
  \end{figure}

\end{example}

\begin{example}
  Below are the three bounds for the grail network in
  Fig.~\ref{fig:GNS_better_than_NS}(a). Since each edge in the grail
  network has a path from it to both destinations or has a path from
  both sources to it, the Network Sharing bound and Cutset bound are
  identical. The one inequality that is different in the two
  collections has been highlighted. As Theorem~\ref{thm:gns-tightest}
  will show, the GNS outer bound collection on the right hand side
  below in fact, contains all possible edge-cut bounds for the grail
  network.

\begin{multicols}{2}
  \begin{center}Cutset Bound and Network Sharing Bound for the Grail
    Network in Fig.~\ref{fig:GNS_better_than_NS}(a)\end{center}
  \vspace{0pt}
        \begin{align*}
          R_1 & \leq C_{e_1} \\
          R_1 & \leq C_{e_4} \\
          R_1 & \leq C_{e_6} \\
          R_1 & \leq C_{e_7} \\
          R_1 & \leq C_{e_9} \\
          R_2 & \leq C_{e_2}+C_{e_3} \\
          R_2 & \leq C_{e_5}+C_{e_8} \\
          R_2 & \leq C_{e_2}+C_{e_8} \\
          \mathbf{R_2} & \mathbf{\leq C_{e_2}+C_{e_7}} \\
          R_1+R_2 & \leq C_{e_1}+C_{e_2}+C_{e_3}\\
          R_1+R_2 & \leq C_{e_3}+C_{e_4}\\
          R_1+R_2 & \leq C_{e_1}+C_{e_2}+C_{e_7}\\
          R_1+R_2 & \leq C_{e_3}+C_{e_5}+C_{e_6}\\
          R_1+R_2 & \leq C_{e_4}+C_{e_7}\\
          R_1+R_2 & \leq C_{e_1}+C_{e_2}+C_{e_8}+C_{e_9}\\
          R_1+R_2 & \leq C_{e_5}+C_{e_7}\\
          R_1+R_2 & \leq C_{e_4}+C_{e_8}+C_{e_9}\\
          R_1+R_2 & \leq C_{e_5}+C_{e_8}+C_{e_9}\\
        \end{align*}

\columnbreak

\begin{center}Generalized Network Sharing (GNS) Bound for the Grail Network
  in Fig.~\ref{fig:GNS_better_than_NS}(a)\end{center}
  \vspace{-20pt}
        \begin{align*}
          R_1 & \leq C_{e_1} \\
          R_1 & \leq C_{e_4} \\
          R_1 & \leq C_{e_6} \\
          R_1 & \leq C_{e_7} \\
          R_1 & \leq C_{e_9} \\
          R_2 & \leq C_{e_2}+C_{e_3} \\
          R_2 & \leq C_{e_5}+C_{e_8} \\
          R_2 & \leq C_{e_2}+C_{e_8} \\
          \mathbf{R_1+R_2} & \mathbf{\leq C_{e_2}+C_{e_7}} \\
          R_1+R_2 & \leq C_{e_1}+C_{e_2}+C_{e_3}\\
          R_1+R_2 & \leq C_{e_3}+C_{e_4}\\
          R_1+R_2 & \leq C_{e_1}+C_{e_2}+C_{e_7}\\
          R_1+R_2 & \leq C_{e_3}+C_{e_5}+C_{e_6}\\
          R_1+R_2 & \leq C_{e_4}+C_{e_7}\\
          R_1+R_2 & \leq C_{e_1}+C_{e_2}+C_{e_8}+C_{e_9}\\
          R_1+R_2 & \leq C_{e_5}+C_{e_7}\\
          R_1+R_2 & \leq C_{e_4}+C_{e_8}+C_{e_9}\\
          R_1+R_2 & \leq C_{e_5}+C_{e_8}+C_{e_9}\\
        \end{align*}
\end{multicols}

\end{example}

\begin{remark}
  The GNS outer bound is a special case of the edge-cut bounds in
  \cite{Kramer06, Harvey06, Thakor09}. However, it has the advantage
  of being simpler and more explicit. Furthermore, from
  Theorem~\ref{thm:gns-tightest}, the GNS outer bound is the tightest
  possible collection of edge-cut bounds for two-unicast networks and
  hence, is equivalent to the bounds in \cite{Kramer06, Harvey06,
    Thakor09} for two-unicast networks. 
\end{remark}

\subsection{Tightness under GNS-cut bottleneck}

The next theorem shows that any minimal GNS-cut, i.e. a GNS-cut with
no proper subset that is also a GNS-cut, provides an outer bound that
is not obviously loose.

\begin{theorem}\label{thm:gns-not-obviously-loose}
  For a given two-unicast graph $\mathcal{G},$ let
  $S\subseteq\mathcal{E}(\mathcal{G})$ be a minimal GNS-cut. Choose an
  arbitrary collection of non-negative reals $\{c_e:e\in S\}.$
  Consider the following link-capacity-vector $\underbar{C}:$
  $C_e=c_e\ \forall e\in S,\ \ \ C_e= \infty \ \forall e\notin S.$
  Then, for the two-unicast network $(\mathcal{G},\underbar{C})$, the GNS
  outer bound is identical to the capacity region
  $\mathcal{C}(\mathcal{G},\underbar{C}),$ i.e. the GNS outer bound is tight.
\end{theorem}

\begin{remark}
  Theorem~\ref{thm:gns-not-obviously-loose} does not say that a sum
  rate of $c_\mathsf{gns}(s_1,s_2;t_1,t_2)=C(S)$ is achievable, only that all rate pairs in
  $\{(R_1,R_2): R_1\leq c(s_1;t_1), R_2\leq c(s_2;t_2), R_1+R_2\leq
  c_\mathsf{gns}(s_1,s_2;t_1,t_2)\}$ are achievable. A sum rate of $C(S)$ is achievable only
  when $C(S)\leq c(s_1;t_1)+c(s_2;t_2)$ for the choice of capacities.
\end{remark}

The proof is relegated to
Appendix~\ref{subsec:proof-gns-not-obvi-loose}. Theorem~\ref{thm:gns-not-obviously-loose}
also holds when $C_e$ for $e\notin S$ are all finite and sufficiently
large, i.e. when $C_e\geq C(S)\ \forall e\notin S.$ This can be
concluded from the proof by using the fact that the coding scheme is
linear over the binary field $\mathbb{F}_2.$

\subsection{The GNS outer bound is not tight}
\label{subsec:GNS-counterexample}

We discussed in Section~\ref{subsec:GNS-tightest-edge-cut} that for
two-unicast networks, the GNS outer bound is equivalent to the bounds
in \cite{Kramer06, Harvey06, Thakor09}. However, the GNS outer bound
is also a special case of the so-called LP bound in \cite{Yeung99},
which is the tighest outer bound obtainable using Shannon information
inequalities alone. In this subsection, we will show that the LP bound
is tighter than the GNS outer bound for general two-unicast networks.
We provide an example of a two-unicast network, the crossfire network
in Fig.~\ref{fig:GNS_counterexample}(a) showing that:
\begin{itemize}
\item the GNS outer bound is not tight, so edge-cut bounds do not
  suffice to characterize the capacity region;
\item the trade-off between rates on the boundary of the capacity
  region need not be 1:1;
\item the capacity region may have a non-integral corner point even if
  all links have integer capacity and thus;
\item scalar linear coding is not sufficient to achieve capacity.
\end{itemize}

Achievability of the capacity region in
Fig.~\ref{fig:GNS_counterexample}(c) follows from a two time step
vector linear coding scheme over $\mathbb{F}_2$ that achieves
$(1,1.5)$ shown in Fig.~\ref{fig:GNS_counterexample}(b).

\begin{figure}[htbp] {\center \subfigure[Crossfire network with all edges having
    unit
    capacity]{\includegraphics[height=3.5in]{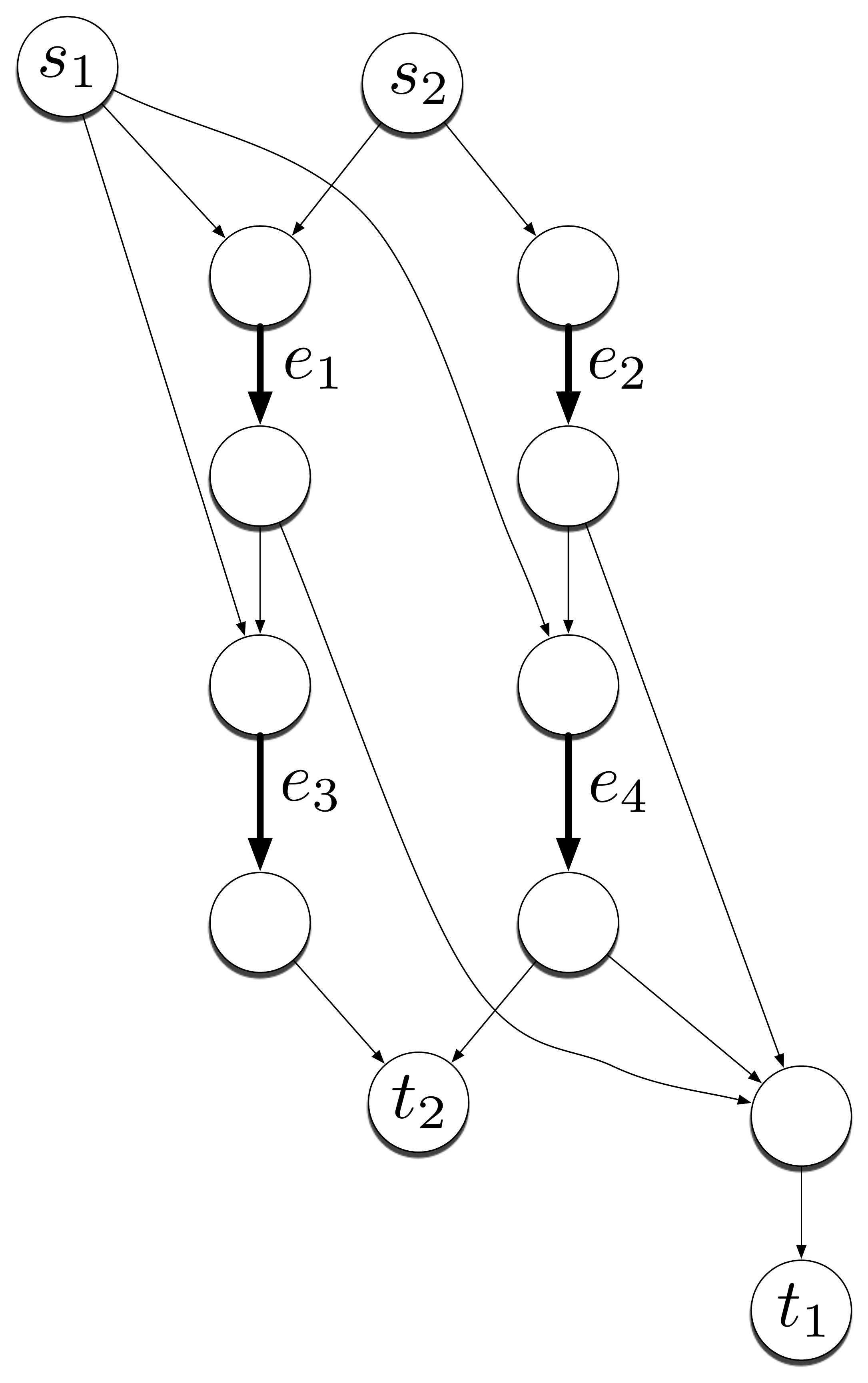}}\hspace{5pt}
    \subfigure[Vector linear scheme over $\mathbb{F}_2$ achieving
    (1,1.5)]{\includegraphics[height=3.5in]{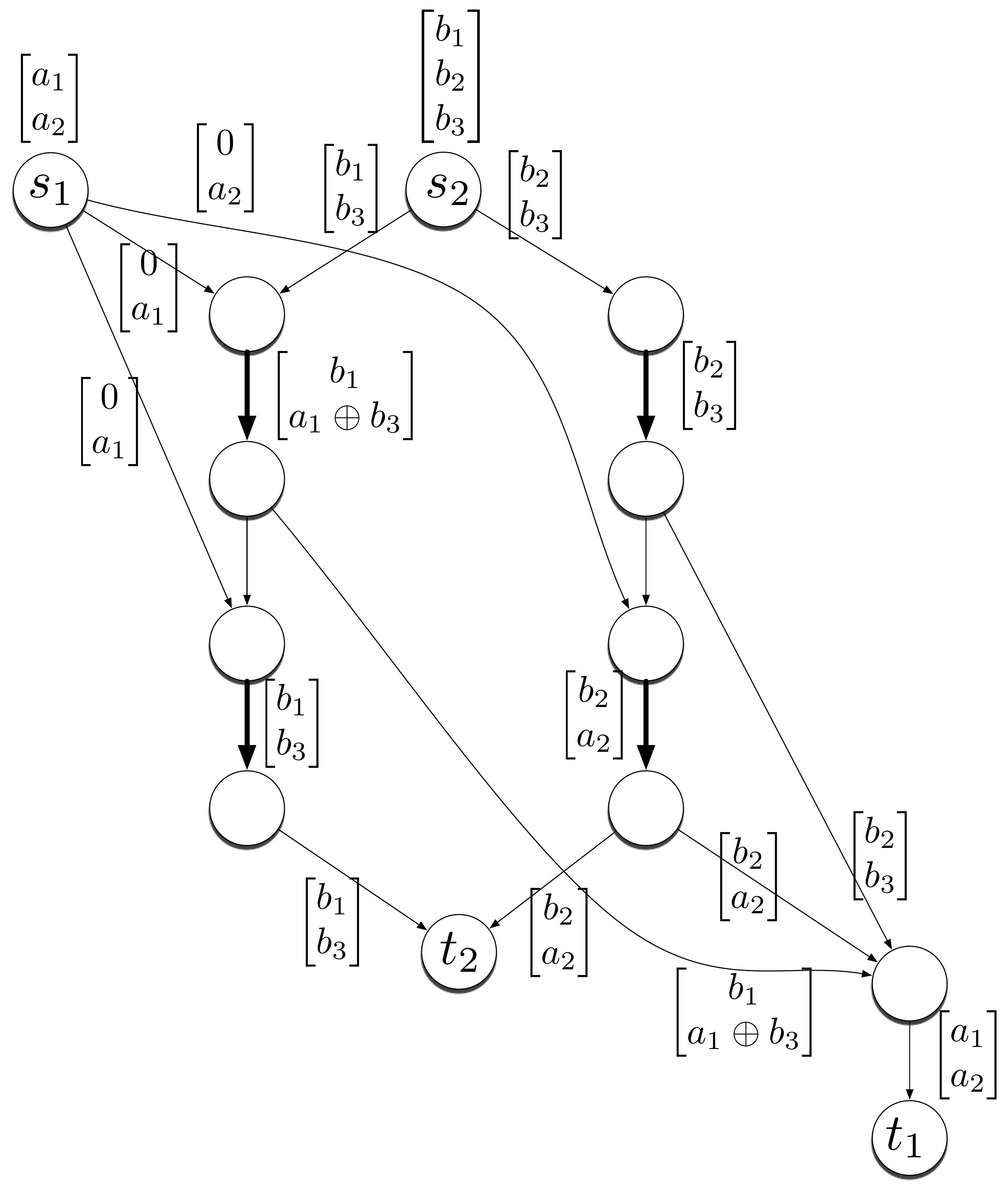}}
    \subfigure[Capacity region of the crossfire network in
    (a)]{\includegraphics[width=2in]{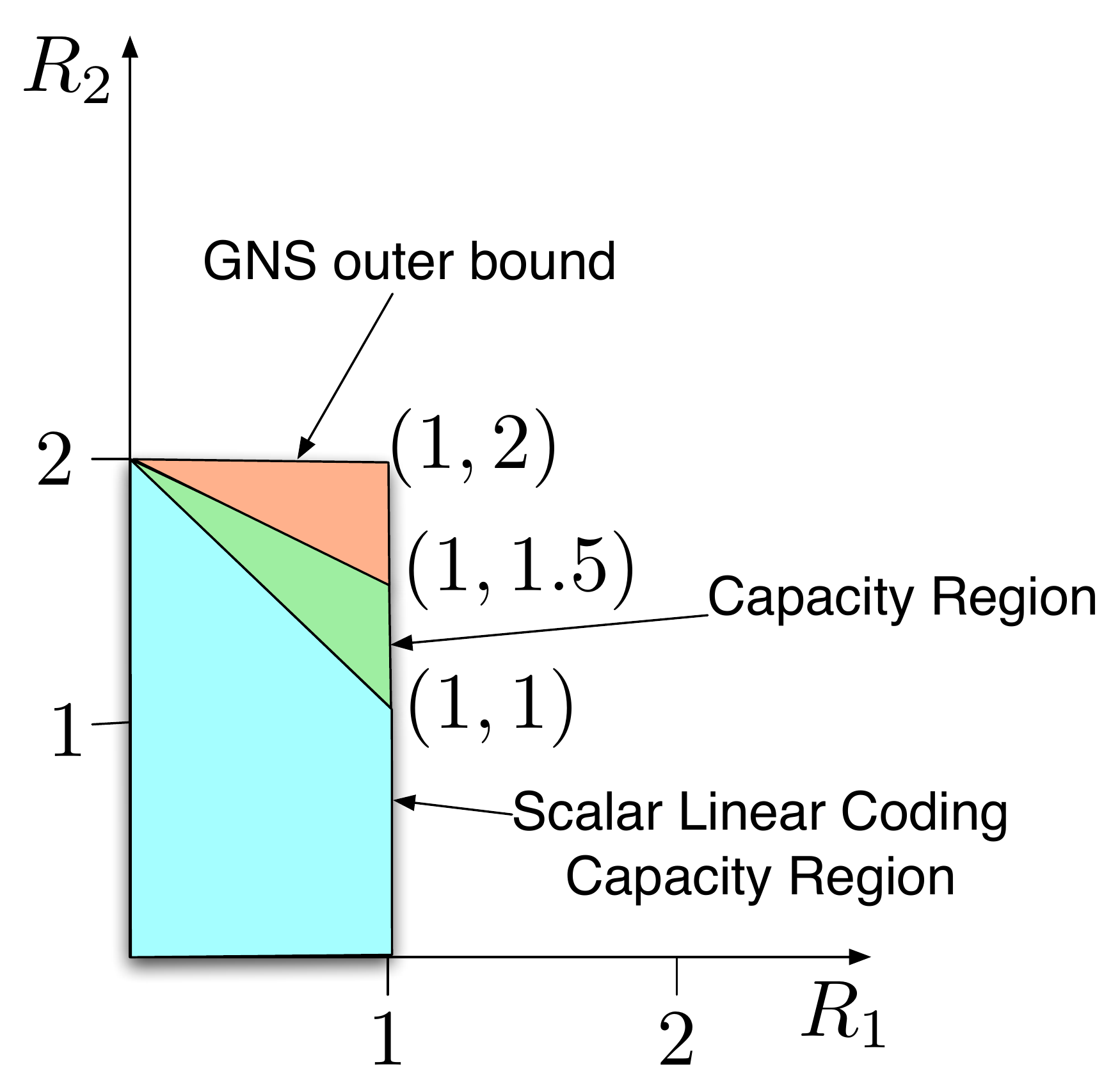}} \caption{GNS
      outer bound is not tight} \label{fig:GNS_counterexample} }
\end{figure}


Suppose $(R_1,R_2)$ is a vanishing error achievable rate pair. Fixing
$0<\epsilon<\frac 12,$ consider a scheme of block length $N,$
achieving the rate pair $(R_1,R_2)$ over alphabet $\mathcal{A}$ with
error probability at most $\epsilon$. Let $W_1,W_2$ be independent and
distributed uniformly over the sets $\mathcal{A}^{\lceil NR_1\rceil}$
and $\mathcal{A}^{\lceil NR_2\rceil}$ respectively. For each edge
$e=e_1,e_2,e_3,e_4,$ define $X_e$ as the concatenated evaluation of
the functions specified by the scheme for edge $e.$ We will assume all
logarithms are to base $|\mathcal{A}|.$

By Fano's inequality,
  \begin{align}
    H(W_1|X_{e_1},X_{e_2},X_{e_4}) & \leq
    h(\epsilon)+\epsilon\lceil NR_1\rceil,\label{eq:one_fano}\\
    H(W_2|W_1,X_{e_1},X_{e_2},X_{e_4}) & \leq h(\epsilon)+\epsilon\lceil NR_2\rceil. \label{eq:two_fano}
  \end{align}
Now, 
  \begin{align}
    NR_1 & \leq H(W_1) \\
    & = I(X_{e_1},X_{e_2},X_{e_4};W_1) +
    H(W_1|X_{e_1},X_{e_2},X_{e_4}) \\
    & = {\color{blue}I(X_{e_1},X_{e_2};W_1)} + {\color{red}I(X_{e_4};W_1|X_{e_1},X_{e_2})} +
    h(\epsilon)+\epsilon\lceil NR_1\rceil \mbox{\ \ \ (from
      \eqref{eq:one_fano})} \label{eq:split-two-terms}\\
    {\color{blue}I(X_{e_1},X_{e_2};W_1)} & =I(X_{e_1},X_{e_2};W_1,W_2) - I(X_{e_1},X_{e_2};W_2|W_1) \\
    & = H(X_{e_1},X_{e_2}) - H(W_2|W_1) +H(W_2|W_1,X_{e_1},X_{e_2}) \\
    & = H(X_{e_1},X_{e_2}) - H(W_2) +h(\epsilon)+\epsilon\lceil
    NR_2\rceil \mbox{\ \ \ \ \ \ \ \ \ \ \ \ \ \ \ \ \ \ (from
      \eqref{eq:two_fano})} \\
    & \leq NC_{e_1}+NC_{e_2} - NR_2 +h(\epsilon)+\epsilon\lceil
    NR_2\rceil \label{eq:first-term}\\    
    {\color{red}I(X_{e_4};W_1|X_{e_1},X_{e_2})} & = I(X_{e_4};W_1,X_{e_1},X_{e_2}) - I(X_{e_4};X_{e_1},X_{e_2}) \\
    & \leq H(X_{e_4}) - I(X_{e_4};W_2) \\
    & = H(X_{e_4}) - I(X_{e_3},X_{e_4};W_2) + I(X_{e_3};W_2|X_{e_4}) \\
    & \leq H(X_{e_4})-H(W_2) + H(X_{e_3}|X_{e_4}) \\
    & = H(X_{e_3},X_{e_4}) - H(W_2) \\
    & \leq NC_{e_3}+NC_{e_4} - NR_2 \label{eq:second-term}
  \end{align}
  From \eqref{eq:split-two-terms}, \eqref{eq:first-term},
  \eqref{eq:second-term}, we can deduce
  \begin{align}
    R_1+2R_2\leq C_{e_1}+C_{e_2}+C_{e_3}+C_{e_4} +
    2(\epsilon+h(\epsilon))+\epsilon(R_1+R_2)
  \end{align}

  Since $\epsilon$ can be made arbitrarily small by a suitable coding
  scheme with a suitable block length, we must have $R_1+R_2\leq
  C_{e_1}+C_{e_2}+C_{e_3}+C_{e_4}=4.$ As this inequality holds for
  every vanishing error achievable rate pair $(R_1,R_2),$ it also
  holds for every point in the closure of the set of vanishing error
  achievable rate pairs. Thus, the network has a capacity region as
  shown in Fig.~\ref{fig:GNS_counterexample}(c).

  Finally, simple routing strategies achieve $(1,1)$ and $(0,2),$
  whereas due to the constraint $R_1+2R_2\leq 4,$ no coding strategy
  can achieve $(1,2).$ As routing is a special case of scalar linear
  coding, the scalar linear coding region is the convex closure of the
  integral rate pairs $(1,1), (0,2), (0,1), (0,0)$ and hence, is as
  shown in Fig.~\ref{fig:GNS_counterexample}(c).

  \subsection{NP-completeness of minimum GNS-cut}
  \label{subsec:np-complete}
  We have shown in Section~\ref{subsec:GNS-tightest-edge-cut} that the
  GNS bound provides the tightest collection of edge-cut bounds for
  two-unicast networks. This brings up a natural question of
  computational complexity of the GNS bound. Since the number of
  GNS-cuts is in general, exponential in the size of the network,
  listing all of them is intractable. For a single-unicast problem, we
  know that there exists an algorithm \cite{FordFulkerson56, Elias56}
  that computes the mincut and reveals a minimizing cut efficiently
  in spite of there being exponentially many edge-cuts. Given a
  two-unicast network, can we find an algorithm that efficiently
  finds, among all GNS-cuts $S,$ one that has the smallest value of
  $\sum_{e\in S} C_e$?  Theorem~\ref{thm:np-completeness} shows
  unfortunately that we cannot (unless P=NP). Define the following
  decision problem:

  \begin{center}\emph{\bf MIN-GNS-CUT}\end{center}

  \begin{center}
    \emph{Instance}: A two-unicast (capacitated) network $\mathcal{N} =
    (\mathcal{G},\underbar{C}).$
  \end{center}
  \begin{center}
    \emph{Question}: Is there a GNS-cut $S$ so that $C(S)\leq K$?
  \end{center}

\begin{theorem}\label{thm:np-completeness}
  MIN-GNS-CUT is NP-complete.
\end{theorem}

\begin{proof}
  It is clear that MIN-GNS-CUT is in NP. We give a polynomial
  transformation from the multiterminal cut problem for three
  terminals which is known to be NP-complete \cite{Dahlhaus94}. In the
  multiterminal cut problem, we are given a number $K$ and an
  unweighted undirected graph $\mathcal{H}$ with three special
  vertices or `terminals' $x,y,z.$ We are asked whether there is a
  subset of edges $F$ of the graph $\mathcal{H}$ with $|F|\leq K$ such
  that $\mathcal{H}\setminus F$ has no paths between any two of
  $x,y,z.$ Given $(\mathcal{H},K),$ we construct a corresponding
  instance of MIN-GNS-CUT as follows. Let the number of edges of
  $\mathcal{H}$ be $N$ with $K\leq N.$

  The two-unicast capacitated network $\mathcal{G}$ is obtained by
  replacing each undirected edge $(u,v)$ of $\mathcal{H}$ with a
  gadget as shown in Fig.~\ref{fig:gadget}. The gadget introduces two
  new vertices $w,w^\prime$ and constitutes five edges, the one
  \emph{central} edge having unit capacity and four \emph{flank} edges
  each having capacity $N+1$ units. Finally, $s_1$ is identified with
  terminal $x,$ $t_2$ with terminal $y$ and both $s_2$ and $t_1$ with
  terminal $z.$

  \begin{figure}[h]
    \begin{center}
      \subfigure[]{\includegraphics[width=1in,height=!]{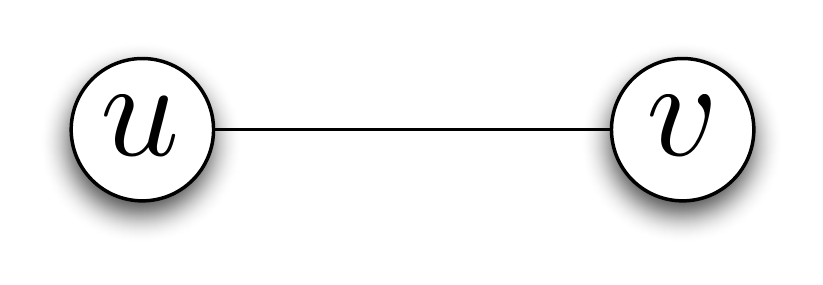}}\hspace{20pt}
      \subfigure[]{\includegraphics[width=1.3in,height=!]{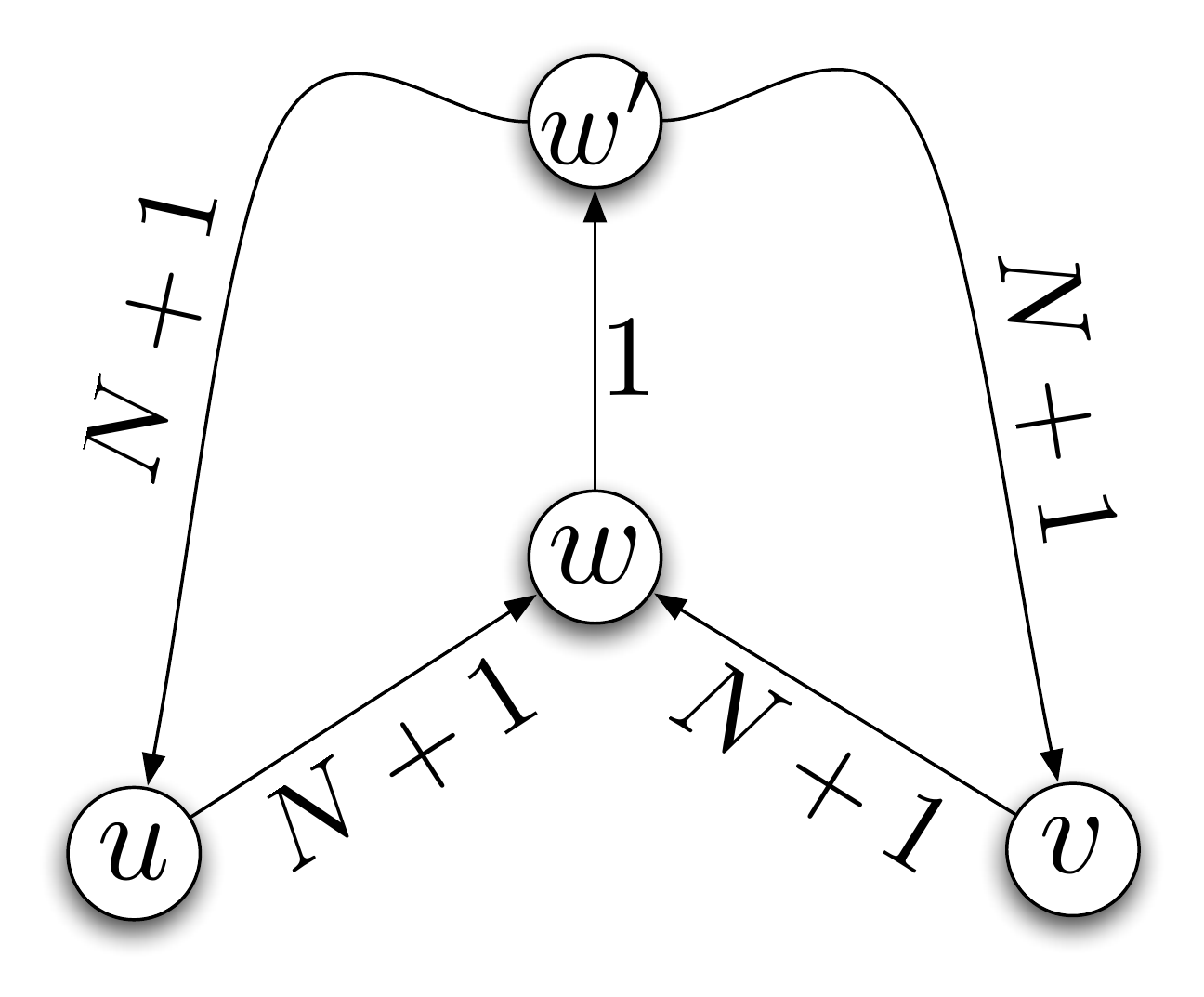}}
      \caption{(a) shows an undirected edge and (b) the corresponding
        gadget}
      \label{fig:gadget}
    \end{center}
  \end{figure}

  We will show that $\mathcal{G}$ has a GNS-cut $S$ with $\sum_{e\in
    S} C_e\leq K$ if and only if $\mathcal{H}$ has a set of edges $F$
  forming a multiterminal cut with $|F|\leq K.$

  Suppose that in the undirected graph $\mathcal{H},$ there is a
  multiterminal cut $F$ with at most $K$ edges. Then, picking the
  central edge of the gadgets corresponding to the edges in $F$ gives
  a GNS-cut $S$ in $\mathcal{G},$ such that $\sum_{e\in S} C_e = |S|
  \leq K.$

  Conversely, suppose there is a GNS-cut $S$ in $\mathcal{G}$ which
  satisfies $\sum_{e\in S} C_e\leq K.$ As $s_2$ and $t_1$ are
  identified, it must be that $\mathcal{G}\setminus S$ has no paths
  from $s_1$ to $t_1,$ from $s_2$ to $t_2$ and from $s_1$ to $t_2.$
  Moreover, as $K\leq N,$ the GNS-cut $S$ cannot contain any flank
  edge, and hence must consist exclusively of central edges of
  gadgets. Choosing the undirected edges of $\mathcal{H}$
  corresponding to the gadgets whose central edges lie in $S$ gives an
  edge set $F$ of $\mathcal{H}$ that has at most $K$ edges and is a
  multiterminal cut in $\mathcal{H}.$
\end{proof}

\subsection{Recent Results}

The Generalized Network Sharing bound has found some interpretations
and extensions in recent years. In this subsection, we summarize some
of these new results:

\begin{itemize}
\item The GNS-cut and GNS bound can be defined analogously for a
  $k$-unicast network (see Theorem~\ref{thm:gns-bound-multiple-unicast})
\item The GNS bound has been given three different interpretations:
  \begin{itemize}
  \item Algebraic interpretation \cite{ZCM12};
  \item Graph-theoretic interpretation via index coding
    \cite{ShanmugamDimakis14};
  \item Network concatenation interpretation \cite{ShomoronyAvestimehr14}.
  \end{itemize}
\item The GNS bound has been observed to be special cases of general
  bounds for a larger family of networks:
  \begin{itemize}
  \item The Generalized Cutset bound for deterministic networks
    \cite{ShomoronyAvestimehr14};
  \item The Chop-and-Roll Directed Cutset bound for general noisy
    networks \cite{KamathKim14}.
  \end{itemize}
\item The GNS-cut for a $k$-unicast network can be approximated to
  within an $O(\log^2 k)$ factor in polynomial time
  \cite{ShanmugamDimakis14}. This has been improved to an
  approximation algorithm that approximates it to within an
  $O(\log k)$ factor in \cite{CKKV15}.
\end{itemize}

\section{Two-unicast is as hard as $k$-unicast}
\label{sec:two-unicast-hard}

A number of recent results have provided evidence that the $k$-unicast
problem for general $k$ is a hard problem \cite{Zeger_insufficiency,
  Zeger_matroid}. We show that in fact, the simplest case $k=2$
encapsulates all the hardness of the general $k$-unicast problem.
Speicifically, we show that the problem of determining whether or not
the rate point $(1,1,\ldots,1)$ is zero-error exactly achievable in a
general $k$-unicast network can be solved if the problem of
determining whether or not the rate point $(k-1,k)$ is zero-error
exactly achievable in a general two-unicast network is solved. For
important technical reasons, we will restrict to zero-error exact
achievability. We will discuss other notions of capacity later in this
section (see Table~\ref{table:reduction} below and also
Remark~\ref{rem:linear-codes} at the end of
Appendix~\ref{subsec:two-unicast-hard}). For simplicity, we will also
restrict to networks with integer link capacities and zero-error
exactly achievable integer rates. These are without loss of
generality. The notion of `hardness' in this section is distinct from
the NP-hardness of computation of GNS-cut that we described in
Sec.~\ref{subsec:np-complete}. We say two-unicast is as hard as
$k$-unicast in the sense that even for two-unicast networks, linear
codes are insufficient for achieving capacity
\cite{Zeger_insufficiency} and Shannon-type information inequalities
are insufficient to characterize capacity \cite{Zeger_matroid}.

\begin{definition}
  For integer rate tuples
  $\underbar{R}=(R_1,R_2,\ldots, R_k),
  \underbar{R}^\prime=(R^\prime_1,R^\prime_2,\ldots, R^\prime_n),$
  we say $\underbar{R}\preceq \underbar{R}^\prime$ if any algorithm or
  computational procedure that can determine whether
  $\underbar{R}^\prime$ is zero-error achievable (or zero-error
  achievable by vector linear coding) in any given network can be used
  to determine whether $\underbar{R}$ is zero-error achievable
  (respectively zero-error achievable by vector linear coding) in any
  given network.
\end{definition}

\begin{figure}[h]
  \centering
  \begin{subfigure}[]{}
    \includegraphics[width=0.3\linewidth]{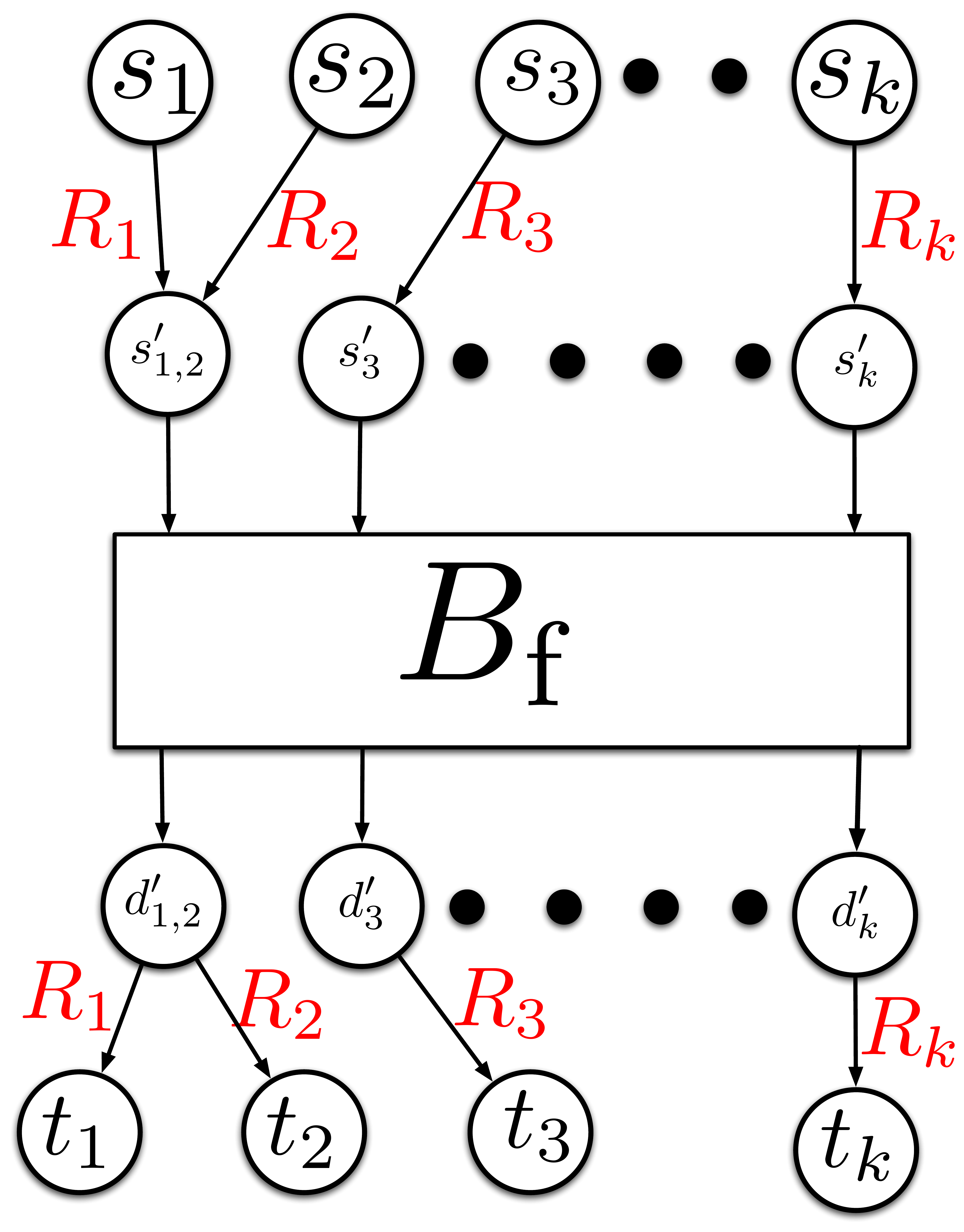}
    \label{fig:fusion}
  \end{subfigure}%
  ~ 
  \hspace{30pt}
  \begin{subfigure}[]{}
    \includegraphics[width=0.3\linewidth]{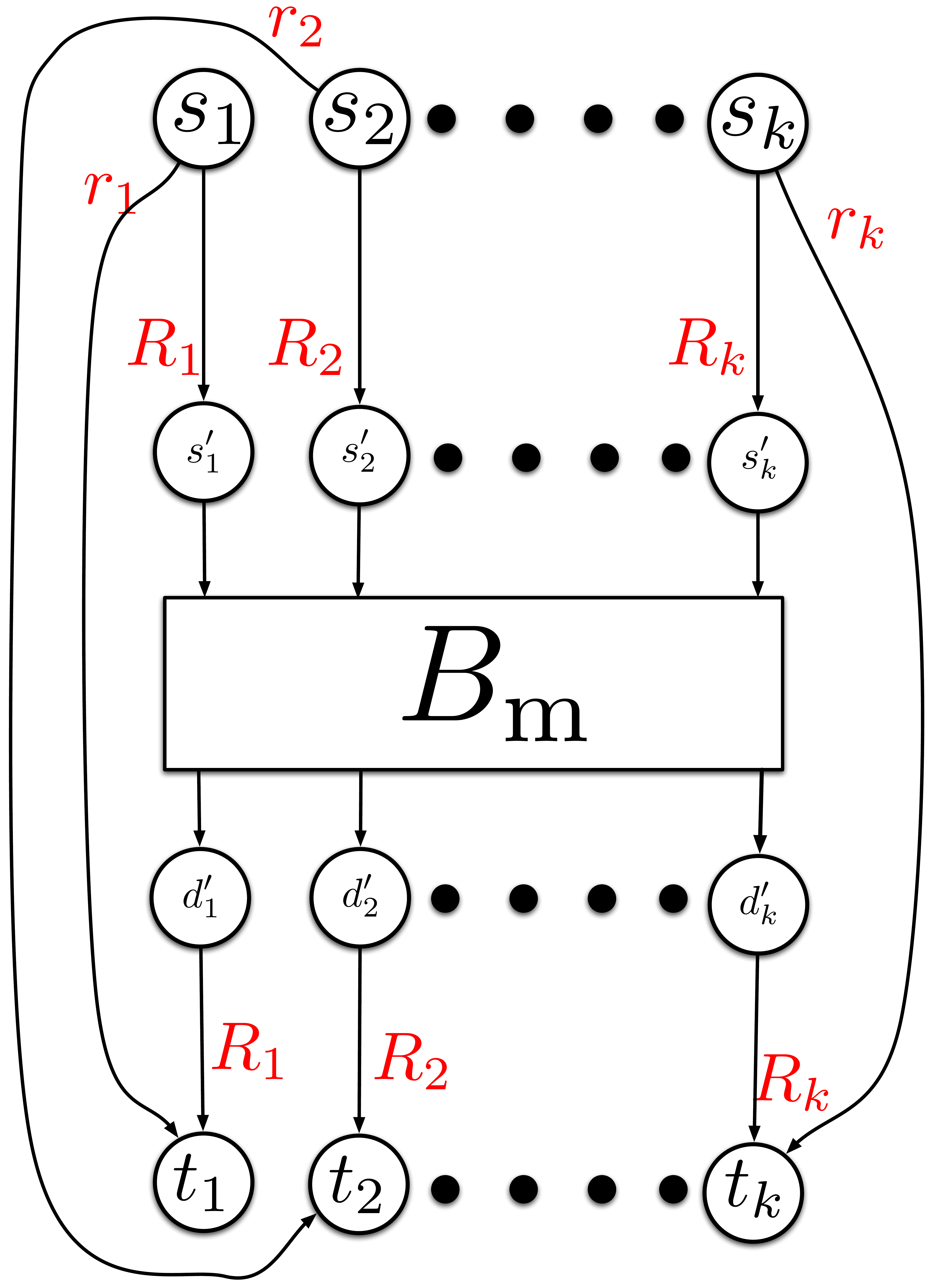}
    \label{fig:monotonicity}
  \end{subfigure}
  \caption{A pictorial proof for the Fusion and Monotonicity
    properties of $\preceq.$ The label in red denotes edge
    capacity. In (a), $(R_1+R_2,R_3,\ldots, R_k)$ is zero-error
    achievable in the network block $B_f$ iff $(R_1,R_2,\ldots, R_k)$
    is zero-error achievable in the extended network. In (b),
    $(R_1,R_2,\ldots, R_k)$ is achievable in the network block $B_m$
    iff $(R_1+r_1,R_2+r_2,\ldots, R_k+r_k)$ is achievable in the
    extended network.}
  \label{fig:fusion_and_monotonicity}
\end{figure}

The following properties may be observed very easily (see
Fig.~\ref{fig:fusion_and_monotonicity}):
\begin{itemize}
\item If $\pi$ is any permutation on $\{1,2,\ldots, k\},$ then
  $(R_1,R_2,\ldots, R_k)\preceq (R_{\pi(1)}, R_{\pi(2)}, \ldots,
  R_{\pi(k)}).$
\item \emph{Fusion:}  $(R_1+R_2,R_3,\ldots, R_k) \preceq (R_1,R_2,R_3,\ldots, R_k).$
\item \emph{Monotonicity:} If $R_i, r_i \geq 0$ for each $i,$ then
  $(R_1,,\ldots, R_k)\preceq (R_1+r_1,\ldots, R_k+r_k).$
\end{itemize}

  \begin{figure}[h]
    \begin{center}
      \includegraphics[width = 3in, height=!]{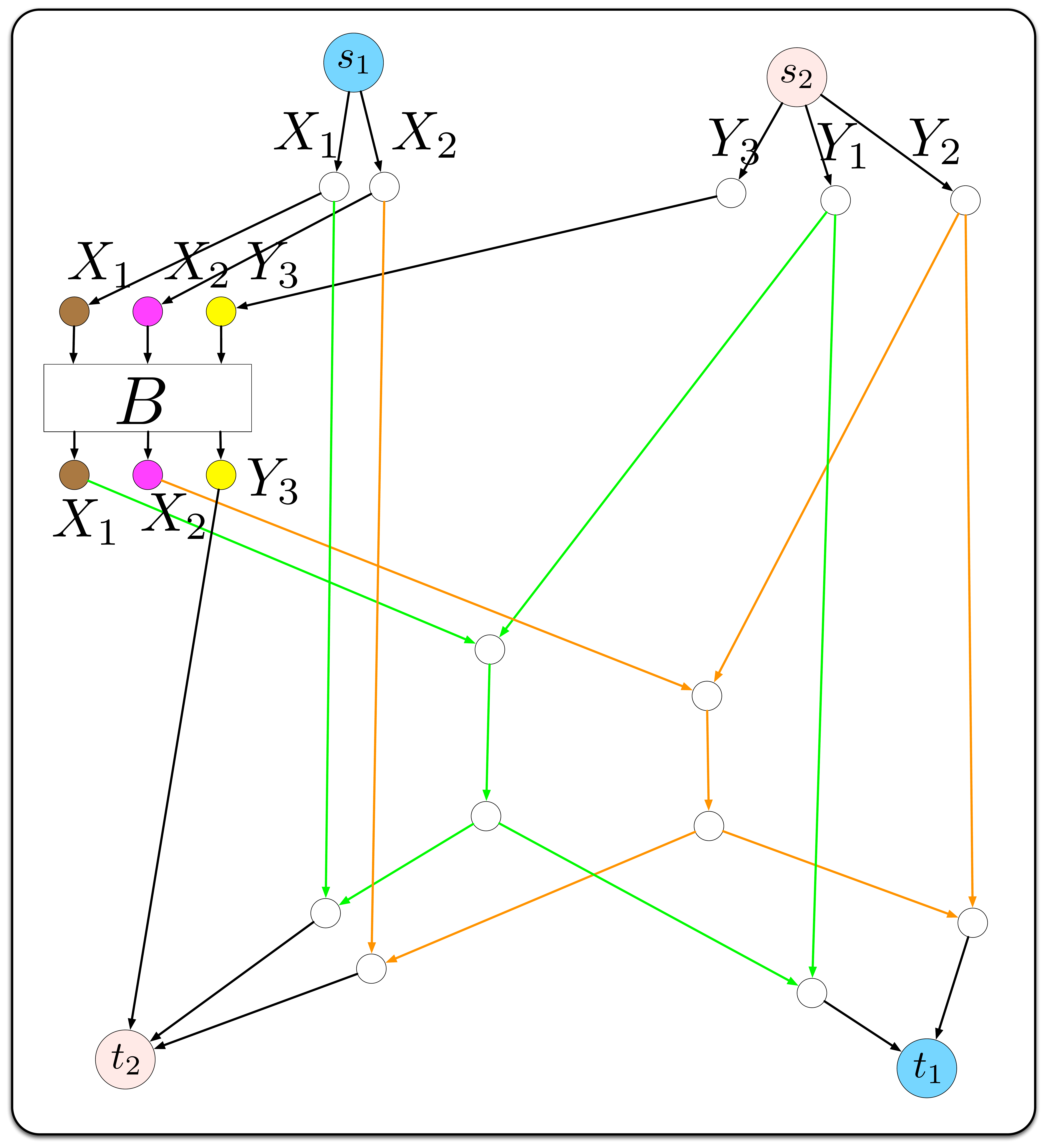}
      \caption{Key idea of the reduction in the proof of
        Theorem~\ref{thm:two-unicast-hard}: (1,1,1) is zero-error
        exactly achievable in the network block $B$ if and only if
        (2,3) is zero-error exactly achievable in the extended
        network}
      \label{fig:two_three_point}
    \end{center}
  \end{figure}

Note that using properties listed above, we can never obtain
$\underbar{R}\preceq \underbar{R}^\prime$ where the number of non-zero
entries in $\underbar{R}^\prime$ is strictly less than that in
$\underbar{R},$ i.e. these properties still suggest that determining
the capacity of a two-unicast network can be strictly easier than
determining that of a $k$-unicast network with $k>2.$ Our result,
Theorem~\ref{thm:two-unicast-hard}, shows that this is not the case.
For pedagogical value, we state and prove the theorem in its full
generality, proposing a reduction from a $(k+m)$-unicast network to an
$(m+1)$-unicast network.

\begin{theorem}
  \label{thm:two-unicast-hard}
  For $k\geq 2, m\geq 1,$ let $R_1, R_2, \ldots, R_k, R_{k+1}, \ldots,
  R_{k+m}, r_1, r_2, \ldots, r_m\geq 0$ be non-negative integers such
  that $\sum_{i=1}^k R_i = \sum_{j=1}^m r_j.$ Then,
  \begin{align}
    (R_1,R_2,\ldots, R_k, R_{k+1}, R_{k+2}, \ldots, R_{k+m}) \preceq
    (\sum_{i=1}^k R_i, R_{k+1}+r_1,R_{k+2}+r_2,\ldots, R_{k+m} + r_m)
  \end{align}
\end{theorem}
\vspace{5pt}

{\large 
  \begin{table*}[ht] 
    {\begin{center}
      \begin{tabularx}{\textwidth}{Z|WV} 
        & & \\
        & \begin{center}Vector Linear codes\end{center} & \begin{center}General codes\end{center}
        \\
        & & \\
        \hline 
        & & \\ 
        \begin{center}Zero-error Exact Achievability\end{center} 

      & \begin{center}\begin{framed}\checkmark\end{framed}\end{center} 

        & \begin{center}\begin{framed}\checkmark\end{framed} \end{center} 

        \\
        & & \\
        \begin{center} Zero-error Capacity \end{center} 

        & \begin{center}\begin{framed}\checkmark\end{framed}\end{center}

        & \begin{center}\begin{framed}?\end{framed}\end{center} 
      
      \\
      & & \\
        \begin{center}Shannon Capacity\end{center} 

      & \begin{center}\begin{framed}\checkmark\end{framed}\end{center}
 
       & \begin{center}\begin{framed}?\end{framed}\end{center} 

        \\
        & & \\

      \end{tabularx}
    \end{center}
  }
  \caption{The notions of capacity for which
    Theorem~\ref{thm:two-unicast-hard} holds are shown by \checkmark
    's. ?'s show the notions for which it is unclear if the proposed
    reduction works.}
    \label{table:reduction}
  \end{table*}
}

The main motivation for Theorem~\ref{thm:two-unicast-hard} is the
following implication:
\begin{align}(R_1,R_2,\ldots, R_k) \preceq \left(\sum_{i=1}^{k-1} R_i,
    \sum_{i=1}^k R_i\right),\end{align} 

i.e. the general two-unicast
problem is as hard as the general multiple-unicast problem. Using the
monotonicity property of $\preceq,$ this suggests that the difficulty
of determining achievability of a rate tuple for the $k$-unicast
problem is related more to the magnitude of the rates in the tuple
rather than the size of $k.$ Moreover, we have
\begin{align}(\underbrace{1,1,1,\ldots, 1}_{\text{$k$ times}})\preceq
  (k-1,k),\end{align} i.e. solving the $k$-unicast problem with
unit-rate (which is known to be a hard problem for large $k$
\cite{Zeger_insufficiency, Zeger_matroid}) is no harder than solving
the two-unicast problem with rates $(k-1,k).$ Furthermore, the
construction in our paper along with the networks in
\cite{Zeger_insufficiency}, \cite{Zeger_matroid},
\cite{Zeger_nonreversibility} can be used to show the
following.

\begin{theorem}\label{thm:insufficiency} There exists a two-unicast network in which a
  non-linear code can achieve the rate pair $(9,10)$ but no linear
  code can.\end{theorem}

\begin{theorem}\label{thm:non-Shannon} There exists a two-unicast network in which non-Shannon
  information inequalities can rule out achievability of the rate pair
  $(5,6),$ but the tighest outer bound obtained using only
  Shannon-type information inequalities cannot.\end{theorem}

We leave the proofs of
Theorems~\ref{thm:two-unicast-hard},~\ref{thm:insufficiency},~\ref{thm:non-Shannon}
to
Appendix~\ref{subsec:two-unicast-hard},~\ref{subsec:insufficiency},~\ref{subsec:non-Shannon}
respectively. We only note here
that the key idea in the proof is a network reduction as shown in
Fig.~\ref{fig:two_three_point}.

We have only considered the notion of zero-error achievability in this
section. We can show that the reduction proposed in the proof of
Theorem~\ref{thm:two-unicast-hard} works for the notions of capacity
as shown in Table~\ref{table:reduction} (see
Remark~\ref{rem:linear-codes} at end of
Appendix~\ref{subsec:two-unicast-hard}).  The reduction works
successfully for the \checkmark 's. It is not known whether the
reduction will work successfully for the notions of capacity given by
?'s.

\section*{Acknowledgements}

VA and SK acknowledge research support during the period this work was
done from the ARO MURI grant W911NF- 08-1-0233, ``Tools for the
Analysis and Design of Complex Multi-Scale Networks'', from the NSF
grant CNS-0910702, and from the NSF Science and Technology Center
grant CCF-0939370, ``Science of Information''. VA also acknowledges
research support from the NSF grant ECCS-1343398, from Marvell
Semiconductor Inc., and from the U.C. Discovery program.

The work of CCW was supported in part by NSF grant CCF-0845968.

\bibliographystyle{IEEE} 
\bibliography{two_unicast_biblio}

\appendix
\subsection{Proof of the Two-Multicast Lemma (Lemma~\ref{lem:two-multicast})}
\label{subsec:lem_two_multicast}

\begin{proof}
  The necessity of these conditions is obvious from the cutset bound
  \cite{ElGamal81}. For proving sufficiency, fix a rate pair
  $(R_1,R_2)$ that satisfies these conditions and consider a new
  network $\tilde{\mathcal{N}}$ obtained by adding a super-source $s$
  with two outgoing edges to $s_1$ and $s_2$ with link capacities
  $R_1$ and $R_2$ respectively as shown in
  Fig.~\ref{fig:two_multicast}. Note that capacities of the newly
  added edges in $\tilde{\mathcal{N}}$ depend on the chosen rate pair
  $(R_1,R_2).$

  \begin{figure}[h]
    \begin{center}
      \includegraphics[width = 3in, height=!]{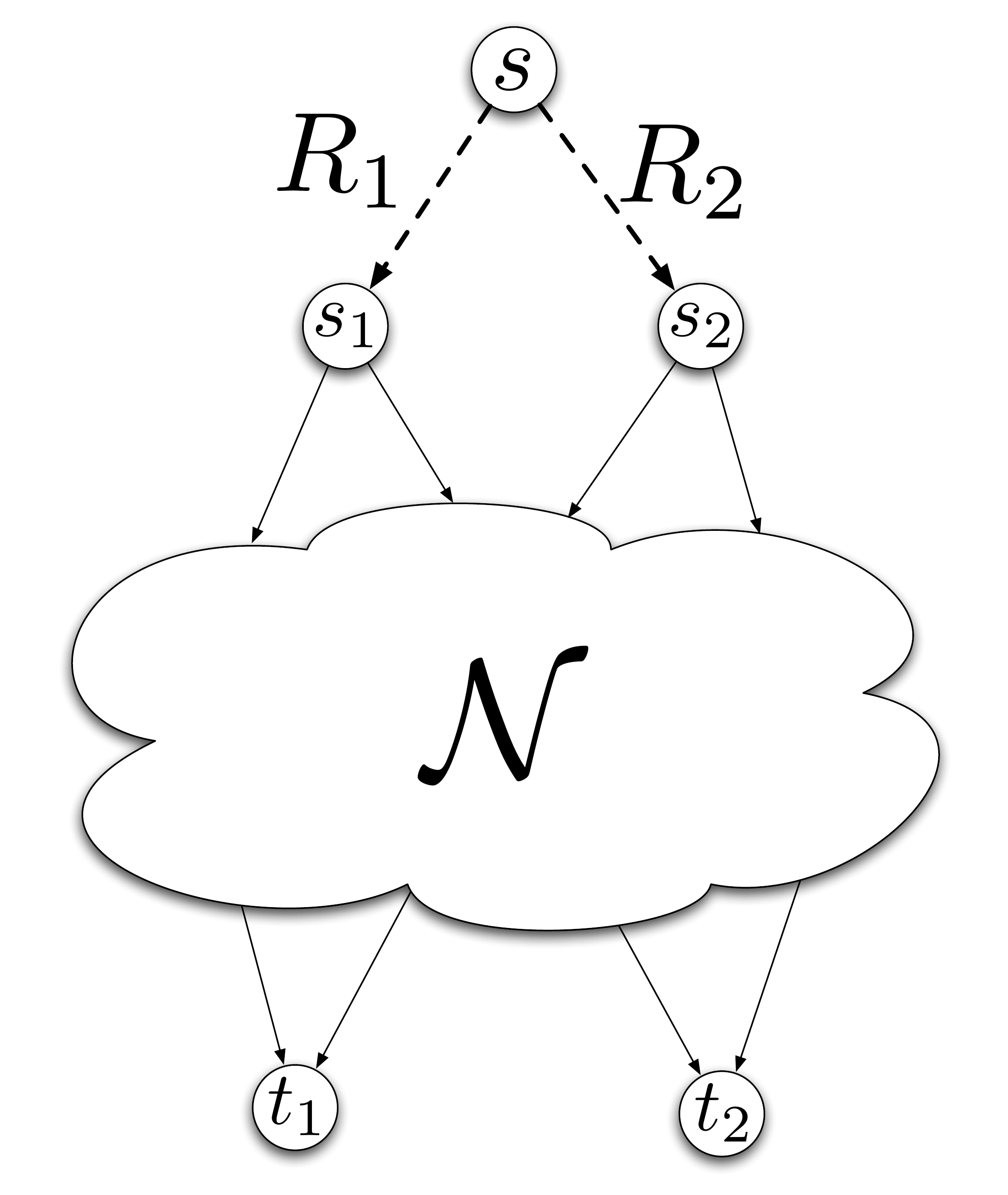}
      \caption{Two-multicast network $\tilde{\mathcal{N}}$ is obtained
        from a given two-unicast network and a specified rate pair $(R_1,R_2).$}
      \label{fig:two_multicast}
    \end{center}
  \end{figure}

  We use the single source multicast theorem \cite{ACLY00},
  \cite{KoetterMedard} on $\tilde{\mathcal{N}}$ to infer the existence
  of a linear coding scheme for super-source $s$ multicasting at rate
  $R_1+R_2$ to the destinations $t_1$ and $t_2.$ This allows us to
  construct a two-multicast scheme in the original network
  $\mathcal{N}$ achieving the desired rate pair.
\end{proof}

\subsection{Proof of Theorem~\ref{thm:gns-not-obviously-loose}}
\label{subsec:proof-gns-not-obvi-loose}
\begin{proof}
  Define $C_i(S):=\min\{C(T): T\subseteq S,\ T \mbox{ is an }s_i-t_i
  \mbox{ cut} \}$ for $i=1,2$ as before.

  As $S$ is a minimal GNS-cut, the GNS outer bound for
  $(\mathcal{G},\underbar{C})$ is given by
  \begin{align}\label{eq:the-gns-bound}
    R_1 \leq C_1(S), R_2 \leq C_2(S), R_1+R_2 \leq C(S).
  \end{align}

  We will assume that $c_e$ is an integer for each $e\in S$ and
  describe scalar linear coding schemes over the binary field
  $\mathbb{F}_2$ with block length $N=1$ achieving the GNS outer
  bound. Having done this, it is easy to see that the theorem would
  also hold for choice of non-negative rational and thus, also
  non-negative real choice of $c_e, e\in S.$ Henceforth, we will
  imagine a link of capacity $c_e$ as having $c_e$ unit capacity edges
  connected in parallel. This change could be made in the graph and in
  this proof, we will use $\mathcal{G}$ to denote the graph with all edges
  having unit capacity, possibly having multiple edges in parallel
  connecting two vertices.

  Note that a given GNS-cut $S$ is minimal if and only if $S\setminus
  e$ is not a GNS-cut for each $e\in S.$ This allows us to partition
  the edges in $S$ by their connectivity in
  $\mathcal{G}\setminus\{S\setminus e\}$ as $S_1^1\cup S_1^2\cup
  S_1^{12}\cup S_2^1\cup S_2^2\cup S_2^{12}\cup S_{12}^1\cup
  S_{12}^2\cup S_{12}^{12}$ where $e\in S$ lies in $S_x^y$ if, in the
  graph $\mathcal{G}\setminus\{S\setminus e\},$ $\mathsf{tail}(e)$ is reachable
  only from source indices $x$ and $\mathsf{head}(e)$ is capable of reaching
  only destination indices $y.$ Eg. $S_{12}^2$ contains edge $e$ in
  $S$ if and only if in $\mathcal{G}\setminus\{S\setminus e\},$ we have
  that $\mathsf{tail}(e)$ is reachable from $s_1,s_2$ and $\mathsf{head}(e)$ can reach
  $t_2,$ but cannot reach $t_1.$

  Define $\hat{S}_1:=S_1^1\cup S_1^{12}\cup S_{12}^1\cup S_{12}^{12}$
  and $\hat{S}_2:=S_2^2\cup S_2^{12}\cup S_{12}^2\cup S_{12}^{12}.$
  Thus, $\hat{S}_i,$ for $i=1,2$ is the set of edges in $S$ which have
  their tails reachable from $s_i$ and their heads reaching $t_i$ by
  paths of infinite capacity. We will show
  $C_i(S)=C(\hat{S}_i)+c_{\mathcal{G}\setminus\hat{S}_i}(s_i;t_i),$ for
  $i=1,2.$ By the Max Flow Min Cut Theorem, there exists a flow of
  value $C_i(S)$ from $s_i$ to $t_i$ in $\mathcal{G}.$ At most
  $C(\hat{S}_i)$ of the flow goes through edges in $\hat{S}_i.$ Thus,
  there exists a flow of value at least $C_i(S)-C(\hat{S}_i)$ in
  $\mathcal{G}\setminus\hat{S}_i.$ So,
  $c_{\mathcal{G}\setminus\hat{S}_i}(s_i;t_i)\geq C_i(S)-C(\hat{S}_i).$
  Now, consider $T_i\subseteq S$ in $\mathcal{G}$ such that $T_i$ is an
  $s_i-t_i$ cut and $C(T_i)=C_i(S).$ Then, since $\hat{S}_i\subseteq
  T_i,$ we have that $T_i\setminus\hat{S}_i$ is an $s_i-t_i$ cut in
  $\mathcal{G}\setminus\hat{S}_i.$ Thus,
  $c_{\mathcal{G}\setminus\hat{S}_i}(s_i;t_i)\leq C(T_i\setminus\hat{S}_i)
  = C(T_i)-C(\hat{S}_i) = C_i(S)-C(\hat{S}_i).$

  \textbf{Case I}: $S$ is a minimal GNS-cut such that $\mathcal{G}\setminus
  S$ has no paths from either of $s_1,s_2$ to $t_1,t_2.$ In this case,
  $S_1^2, S_2^1=\emptyset$ by minimality of $S.$ Thus,
  $C_1(S)+C_2(S)\geq C(\hat{S}_1)+C(\hat{S}_2) =
  C(S)+C(S_{12}^{12})\geq C(S).$ So, in this case, the GNS outer bound
  \eqref{eq:the-gns-bound} is a pentagonal region and we have to show
  achievability of the two corner points $(C_1(S), C(S)-C_1(S))$ and
  $(C(S)-C_2(S),C_2(S)).$

  Consider the following scheme. Edges in
  $S_1^1,S_1^{12},S_{12}^1,S_{12}^{12}$ forward $s_1$'s message bits
  to $t_1$ and edges in $S_2^2, S_{12}^2, S_2^{12}$ forward $s_2$'s
  message bits to $t_2.$ This achieves {\small
    \begin{align*}
      R_1 &= C(\hat{S}_1) = C(S_1^1)+C(S_1^{12})+C(S_{12}^1)+C(S_{12}^{12}), \\
      R_2 &= C(S_2^2)+C(S_{12}^2)+C(S_2^{12}).
    \end{align*} } Note that we have $R_1+R_2=C(S)$ for this rate
  pair. Now, we will increase $R_1$ up to $C_1(S)$ while preserving
  this sum rate. Construct $c_{\mathcal{G}\setminus\hat{S}_1}(s_1;t_1)$
  unit capacity edge-disjoint paths from $s_1$ to $t_1$ in
  $\mathcal{G}\setminus\hat{S}_1.$ This gives us
  $c_{\mathcal{G}\setminus\hat{S}_1}(s_1;t_1)$ paths in $\mathcal{G}$ such that
  none of them use any edge in $\hat{S}_1.$ Any such path encounters a
  first finite capacity edge from $S_{12}^2$ and a last finite
  capacity edge from $S_2^{12}.$ The intermediate finite capacity
  edges may be assumed to lie in $S_2^2$ only. If intermediate finite
  capacity edges lie in $S_{12}^2$ or $S_{2}^{12},$ we can modify the
  path so that this is not the case, while preserving the
  edge-disjointness property.
  Now, a simple XOR coding scheme as shown in
  Fig.~\ref{fig:GNS_achievability2}(a) improves $R_1$ by one bit and
  reduces $R_2$ by one bit as $s_2$ has to set $b_1\oplus b_2\oplus
  b_3=0$ to allow $t_1$ to decode $a.$ In the general case, we have an
  arbitrary number of finite capacity edges from $S_2^2$ along the
  path, for which we perform a similar XOR scheme. Because the paths
  are edge-disjoint, the finite capacity edges on those paths are all
  distinct, so the imposed constraints can all be met by reducing
  $R_2$ by one bit for each such path. When this is carried out for
  each of the $c_{\mathcal{G}\setminus\hat{S}_1}(s_1;t_1)$ paths, we have a
  scheme achieving $(C_1(S), C(S)-C_1(S)).$ Similarly,
  $(C(S)-C_2(S),C_2(S))$ may be shown to be achievable.


  \textbf{Case II}: $S$ is a minimal GNS-cut such that
  $\mathcal{G}\setminus S$ has no paths from $s_1$ to $t_1,$ $s_2$ to
  $t_2,$ or $s_2$ to $t_1$ but it has paths from $s_1$ to $t_2.$ As
  $S$ is a minimal GNS-cut, we have $S_1^2=\emptyset.$ In this case,
  the GNS outer bound \eqref{eq:the-gns-bound} is not necessarily a
  pentagonal region. We first show achievability of the rate pair
  $R_1=C_1(S), R_2=\min\{C_2(S), C(S)-C_1(S)\}.$

  \textbf{Stage I - Basic Scheme}: It is easy to see that we can
  achieve the rate pair given by {\small
    \begin{align*}
      R_1 &= C(\hat{S}_1) = C(S_1^1)+C(S_{12}^1)+C(S_1^{12})+C(S_{12}^{12}), \\
      R_2 &=
      C(S_2^2)+C(S_2^{12})+C(S_{12}^2)+\min\{C(S_2^1),C(S_{12}^{12})\},
    \end{align*}
  } by a routing + butterfly coding approach as follows.
  \begin{itemize}
  \item Edges in $S_1^1, S_1^{12}, S_{12}^1$ forward $s_1$'s message
    bits to $t_1$ and edges in $S_2^2, S_2^{12}, S_{12}^2$ forward
    $s_2$'s message bits to $t_2.$
  \item Edges in $S_{12}^{12}$ and $S_2^1$ along with an infinite
    capacity path from $s_1$ to $t_2$ perform ``preferential routing
    for $s_1$ with butterfly coding for $s_2,$'' i.e.
    \begin{itemize}
    \item if $C(S_2^1)<C(S_{12}^{12}),$ then an amount of
      $C(S_{12}^{12})-C(S_2^1)$ of the capacity of edges in
      $S_{12}^{12}$ is used for routing $s_1$'s message bits, while
      the rest is used for butterfly coding, i.e. an XOR operation is
      performed over $C(S_2^1)$ bits from source $s_1$ with $C(S_2^1)$
      bits from source $s_2$ to be transmitted over the edges in
      $S_{12}^{12}.$ Edges in $S_2^1$ provide $C(S_2^1)$ bits of
      side-information from $s_2$ to $t_1,$ while the infinite
      capacity path from $s_1$ to $t_2$ provides side-information to
      $t_2.$
    \item if $C(S_2^1)\geq C(S_{12}^{12}),$ then all of the capacity
      of edges in $S_{12}^{12}$ is used for butterfly coding.
    \end{itemize}
  \end{itemize}

  \textbf{Stage II - Improving $R_1$ up to $C_1(S)$}: We know
  $c_{\mathcal{G}\setminus\hat{S}_1}(s_1;t_1)=C_1(S)-C(\hat{S}_1).$ Find
  $c_{\mathcal{G}\setminus\hat{S}_1}(s_1;t_1)$ unit capacity edge-disjoint
  paths from $s_1$ to $t_1$ in $\mathcal{G}$ such that none of them use any
  edge in $\hat{S}_1.$ Each such unit capacity path from $s_1$ to
  $t_1$ in $\mathcal{G}$ starts with a first finite capacity edge in
  $S_{12}^2,$ ends with the last finite capacity edge in $S_2^{12}$ or
  $S_2^1$ and with all intermediate edges lying, without loss of
  generality, in $S_2^2.$ Whenever the capacity of all edges in
  $S_2^1$ is used up, we would have reached a sum rate of $C(S),$ as
  all edges are carrying independent linear combinations of message
  bits. In that case, we will increase $R_1$ by one bit and reduce
  $R_2$ by one bit. Else, we will increase $R_1$ by one bit while not
  altering $R_2.$
  \begin{itemize}
  \item If the last finite capacity edge lies in $S_2^{12},$ perform
    coding as in Fig.~\ref{fig:GNS_achievability2}(a). If the capacity
    of $S_2^1$ edges is not fully used, use free unit capacity of some
    edge $e\in S_2^1$ to relay the XOR value of $b_1\oplus b_2\oplus
    b_3$ from $s_2$ to $t_1.$ Use the infinite capacity path from
    $s_1$ to $t_2$ to send the symbol $a.$ If there is no free edge in
    $S_2^1,$ then $s_2$ sets $b_1\oplus b_2\oplus b_3=0.$ This
    increases $R_1$ by one bit and reduces $R_2$ by one bit.

    \begin{figure}[htbp]
      {\center \subfigure[Coding Performed in Case I. Also used in
        Case II, Stage II - Last finite capacity edge in
        $S_2^{12}$]{\includegraphics[width=2in]{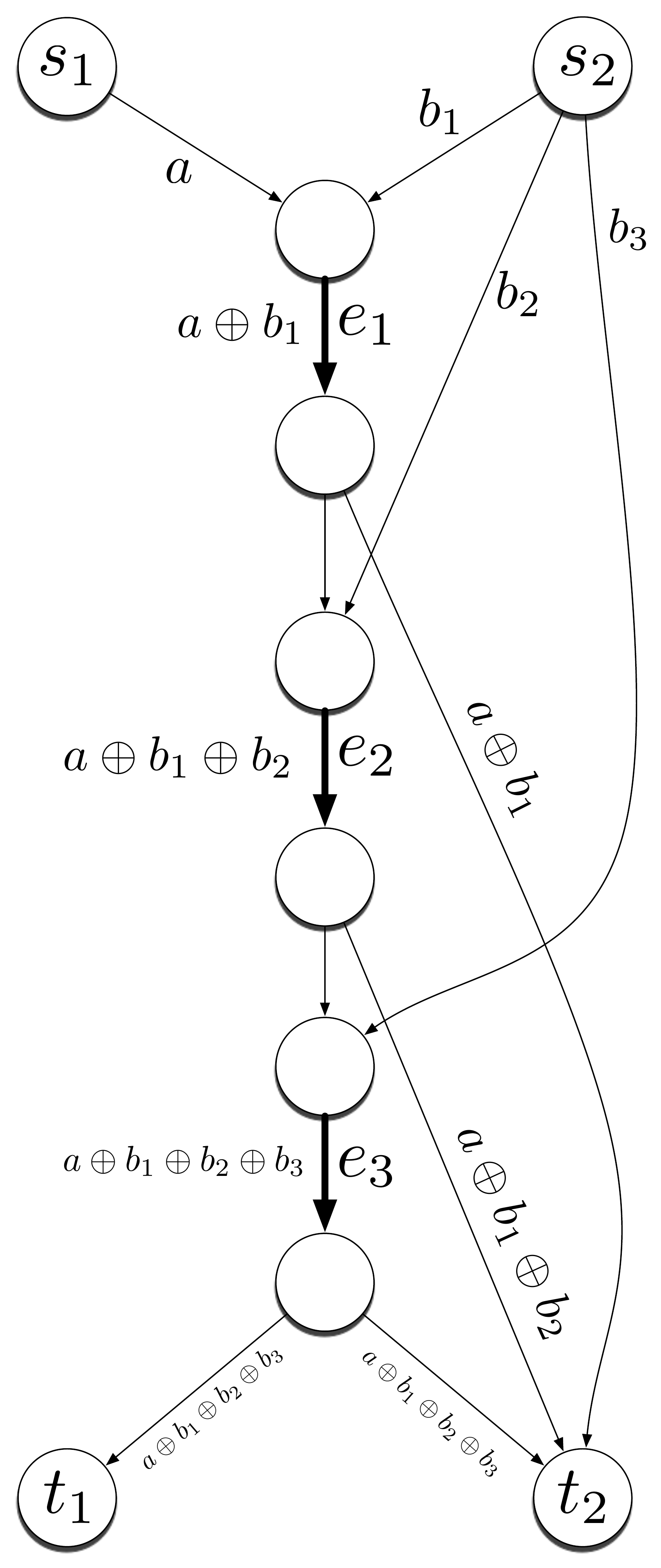}}\hspace{15pt}
        \subfigure[Case II, Stage II - Last finite capacity edge in
        $S_2^1$]{\includegraphics[width=2in]{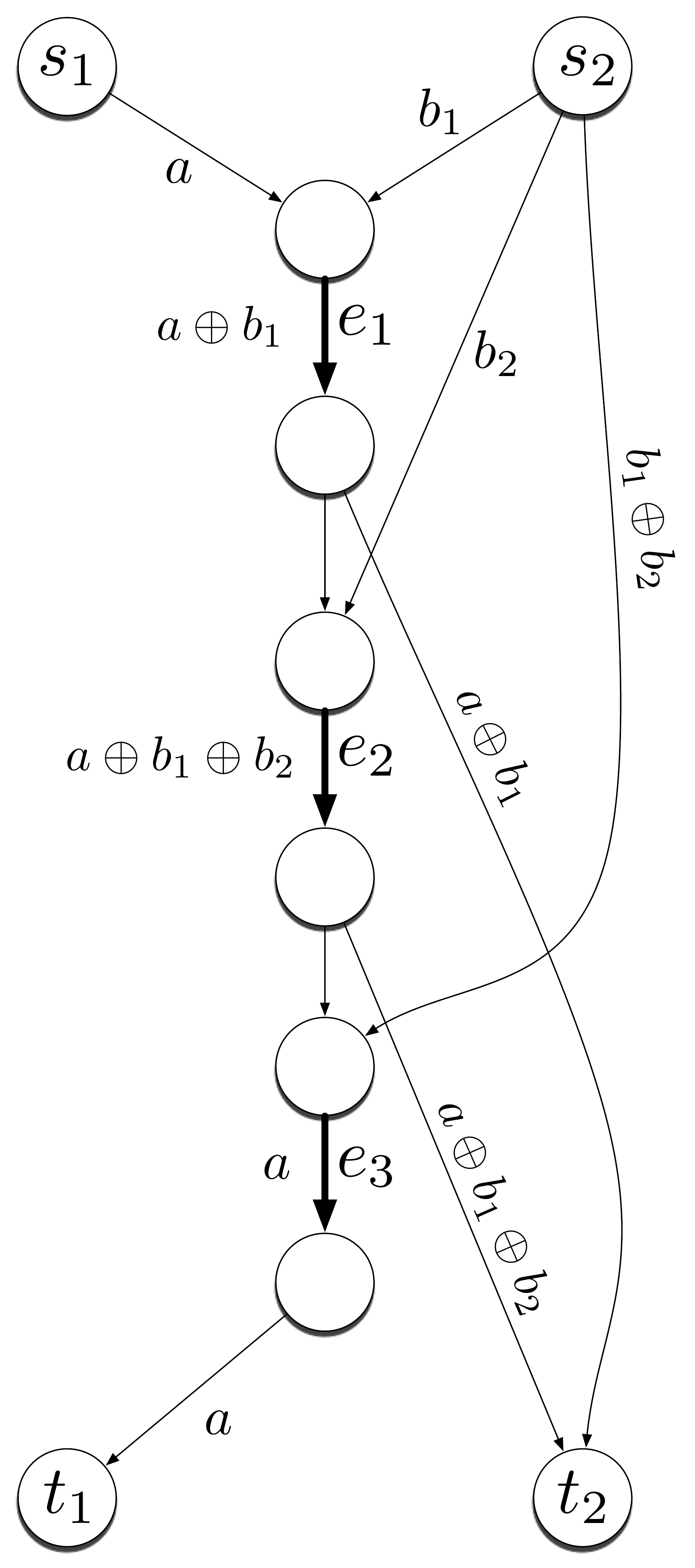}}
        \caption{Improving $R_1$ up to $C_1(S)$}
        \label{fig:GNS_achievability2}
      }
    \end{figure}

  \item Suppose the last finite capacity edge, call it $e_3,$ lies in
    $S_2^1.$ Suppose there is a free edge $e\in S_2^1.$ If $e_3$ is
    being used, it must be used as a conduit for side-information to
    $t_1,$ as part of the butterfly coding. Use $e$ to relay that
    side-information to $t_1.$ So, we can assume $e_3$ is free. Now,
    perform coding as in Fig.~\ref{fig:GNS_achievability2}(b). Use the
    infinite capacity path from $s_1$ to $t_2$ to relay the symbol
    $a.$ This improves $R_1$ by one bit while $R_2$ remains
    unchanged. If there is no free edge in $S_2^1,$ then we must have
    achieved a sum rate of $C(S).$ Edge $e_3$ now relays $a$ to $t_1$
    improving $R_1$ by one bit. However, the edge $e_3$ must have been
    assisting in butterfly coding using some edge in $S_{12}^{12}$ and
    the infinite capacity $s_1-t_2$ path. Now, the edge $e_3$ can no
    longer provide side-information to $t_1.$ So, the corresponding
    unit capacity in some edge in $S_{12}^{12}$ now performs routing
    of $s_1$'s message bit as opposed to XOR mixing of one bit of
    $s_1$'s message and one bit of $s_2$'s message. This reduces $R_2$
    by one bit.
  \end{itemize}

  This can be carried out for the $c_{\mathcal{G}\setminus
    \hat{S}_1}(s_1;t_1)$ edge-disjoint paths sequentially.

  \textbf{Stage III - Improving $R_2$ up to
    $\min\{C(S)-C_1(S),C_2(S)\}$}: If the capacity of $S_2^1$ edges is
  all used up, we have achieved a sum rate of $R_1+R_2=C(S)$ and so,
  $R_2=C(S)-C_1(S).$ If not, we have $R_1=C_1(S), R_2=C(\hat{S}_2).$
  We have $C_2(S)=C(\hat{S}_2)+c_{\mathcal{G}\setminus
    \hat{S}_2}(s_2;t_2).$ Similar to before, we find
  $c_{\mathcal{G}\setminus\hat{S}_2}(s_2;t_2)$ unit capacity edge-disjoint
  paths from $s_2$ to $t_2$ in $\mathcal{G}$ such that the paths don't use
  any edge in $\hat{S}_2.$ Each such unit capacity path encounters a
  first finite capacity edge from $S_{12}^1$ or $S_2^1$ and a last
  finite capacity edge from $S_1^{12}$ while all intermediate finite
  capacity edges may be assumed to lie in $S_1^1.$ Note that edges in
  $S_1^1, S_1^{12}, S_{12}^1$ are all performing pure routing of
  $s_1$'s message. At any point, if the capacity of $S_2^1$ edges is
  fully used, we have reached $R_1=C_1(S), R_2=C(S)-C_1(S).$ If the
  capacity is not fully used, perform the modification as described
  below.

  \begin{itemize}
  \item If the first finite capacity edge lies in $S_{12}^1,$ perform
    coding as in Fig.~\ref{fig:GNS_achievability3}(a).  Use unit
    capacity of a free edge in $S_2^1$ to relay symbol $b$ from $s_2$
    to $t_1$ and use the $s_1$ to $t_2$ infinite capacity path to send
    the XOR value of $a_1\oplus a_2\oplus a_3$ to $t_2.$ This leaves
    $R_1$ unaffected and improves $R_2$ by one bit.

    \begin{figure}[htbp]
      {\center \subfigure[Case II, Stage III - First finite capacity
        edge in
        $S_{12}^1$]{\includegraphics[width=2in]{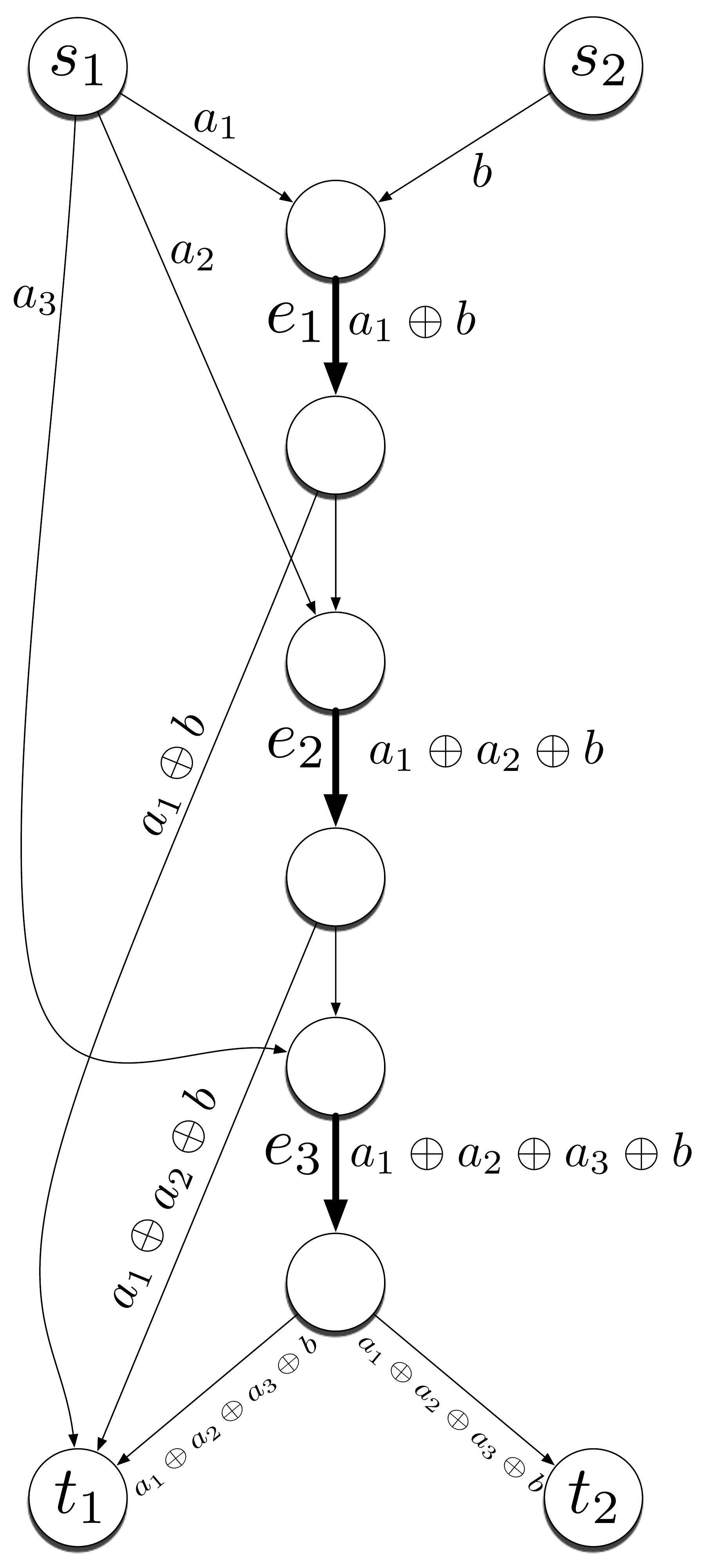}}\hspace{15pt}
        \subfigure[Case II, Stage III - First finite capacity edge in
        $S_2^1$]{\includegraphics[width=2in]{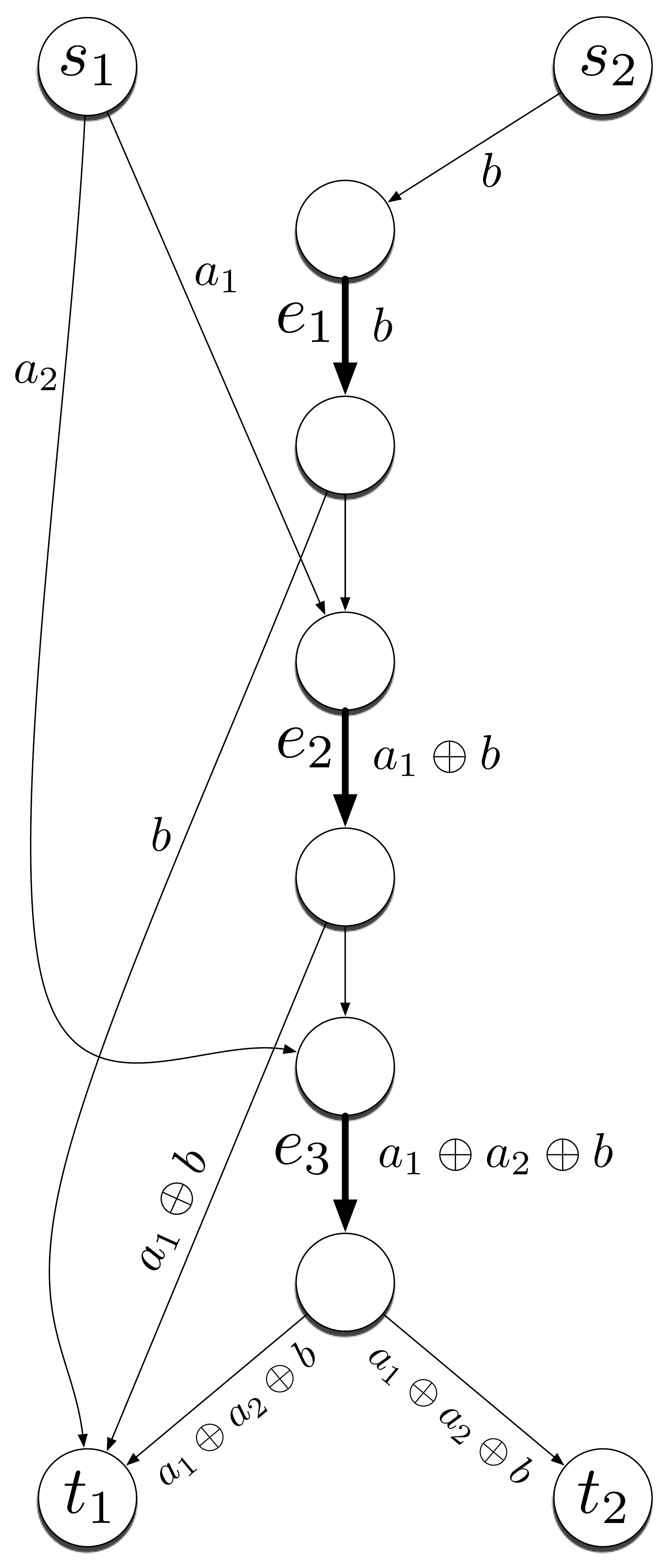}}
        \caption{Improving $R_2$ up to $\min\{C(S)-C_1(S),C_2(S)\}$}
        \label{fig:GNS_achievability3}
      }
    \end{figure}

    \begin{figure}[htbp]
      {\center \subfigure[Case II, Stage III - $e^\prime,
        e^{\prime\prime}, e_1$ are being used in Stage II. $e_2, e_3$
        serve to route $s_1$'s bits to
        $t_1.$]{\includegraphics[width=2.5in]{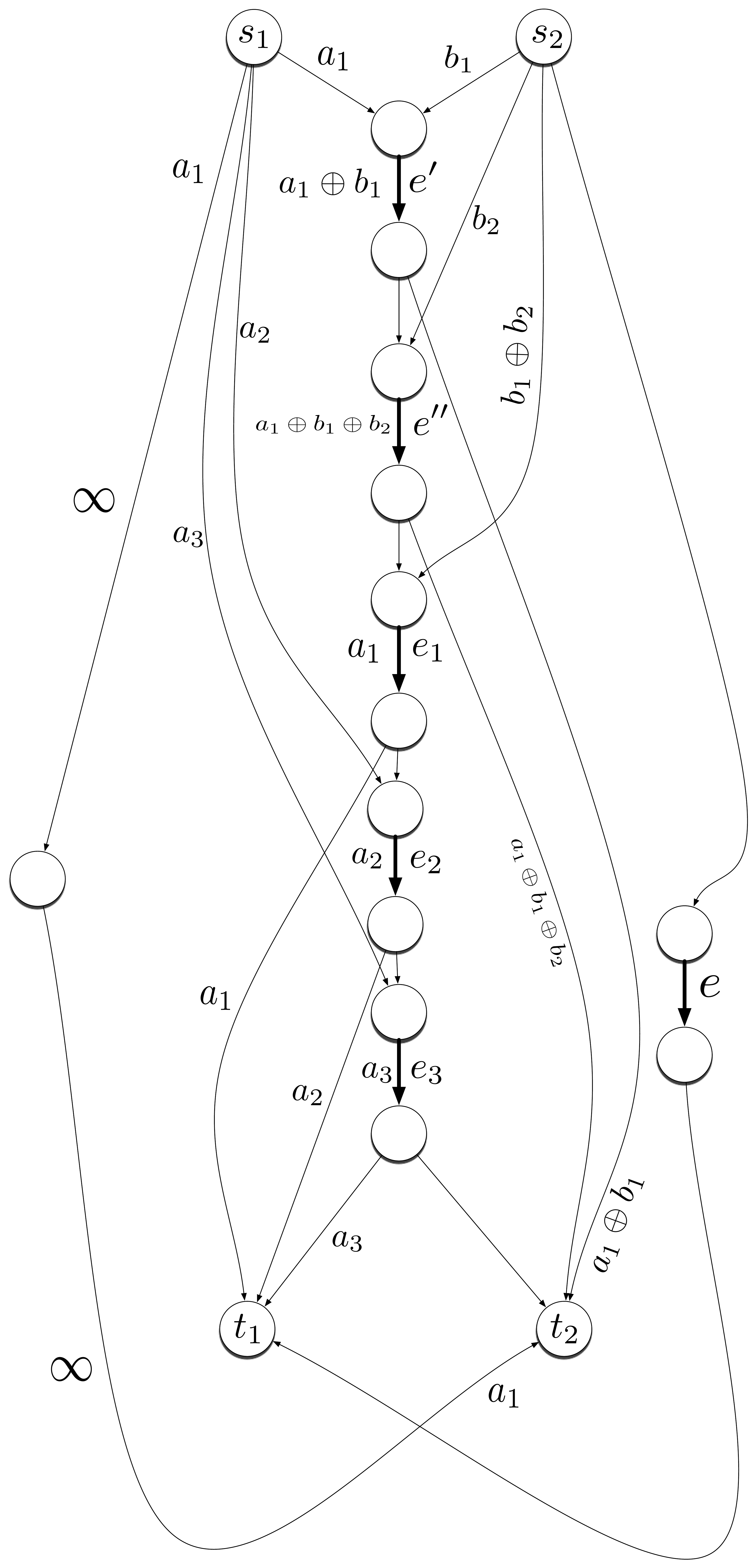}}
        \subfigure[Case II, Stage III - Chosen $s_2$-$t_2$ path uses
        edges $e_1,e_2,e_3$. Modified scheme uses some free edge $e\in
        S_2^1.$]{\includegraphics[width=2.5in]{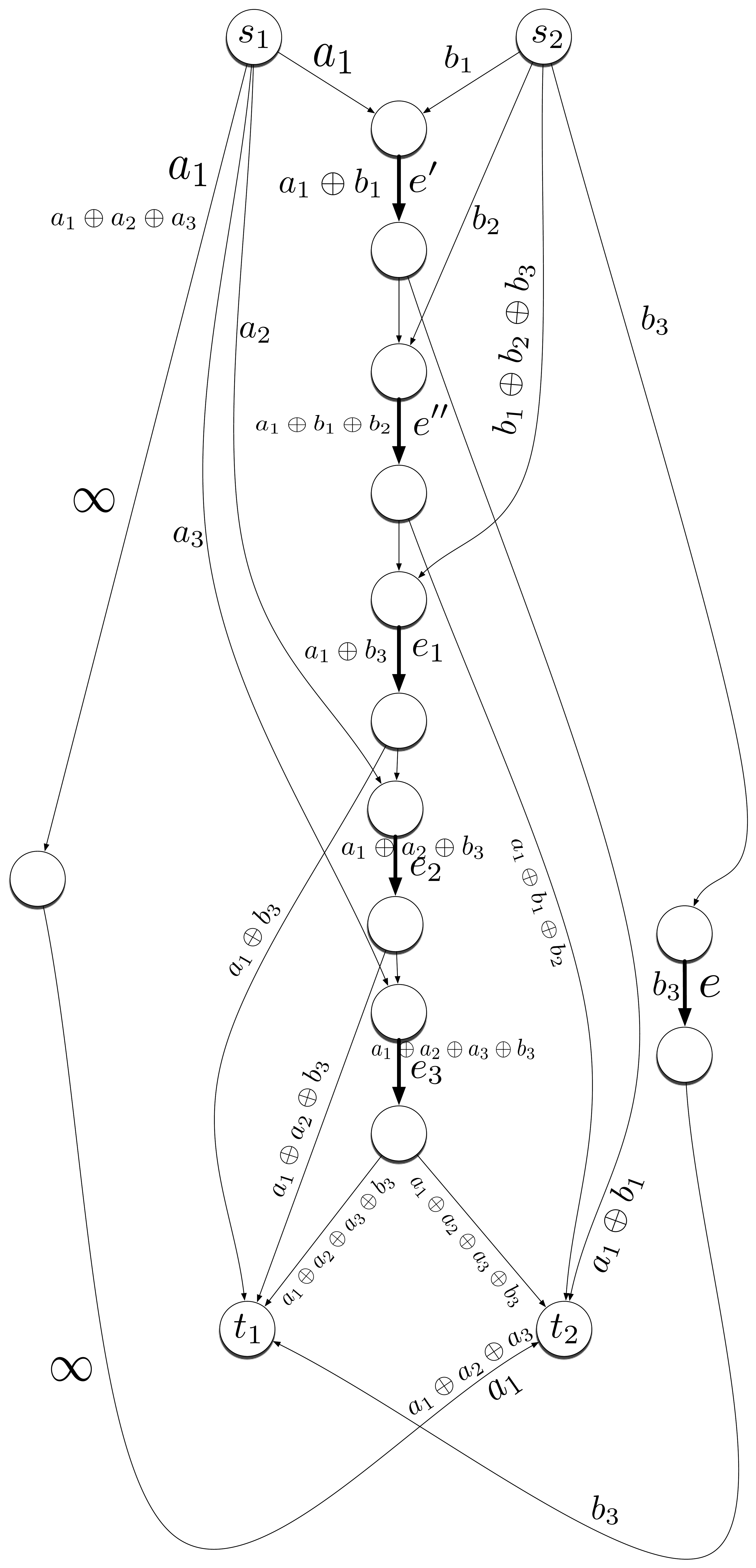}}
        \caption{Improving $R_2$ up to $\min\{C(S)-C_1(S),C_2(S)\}$ in
          the case when $e_1$ was already being used in Stage II.}
        \label{fig:GNS_achievability3_new}
      }
    \end{figure}

  \item Suppose the first finite capacity edge, call it $e_1,$ lies in
    $S_2^1.$ If $e_1$ is not being used, perform coding as in
    Fig.~\ref{fig:GNS_achievability3}(b). Use unit capacity of edge
    $e_1\in S_2^1$ to send a symbol $b$ from $s_2$ to $t_1.$ The
    infinite capacity $s_1$ to $t_2$ path is used to send $a_1\oplus
    a_2$ from $s_1$ to $t_2.$ This allows $t_2$ to decode $b$ and
    improves $R_2$ by one bit while leaving $R_1$ unaffected. If $e_1$
    is being used for sending side-information to $t_1$ (as part of
    the butterfly coding or in Stage II), then pick some free edge
    $e\in S_2^1$ for the transfer of side-information freeing up $e_1$
    and allowing us to use the coding described in
    Fig.~\ref{fig:GNS_achievability3}(b). If $e_1$ is being used but
    not for sending side-information, it must have gotten used in
    Stage II as the last finite capacity edge on an $s_1-t_1$ path. In
    this case, we use some free edge $e\in S_2^1$ and superimpose
    scheme shown in Fig.~\ref{fig:GNS_achievability3}(b) with already
    existing scheme Fig.~\ref{fig:GNS_achievability2}(b). This
    modification is shown via Fig.~\ref{fig:GNS_achievability3_new}(a)
    and Fig.~\ref{fig:GNS_achievability3_new}(b). This improves $R_2$
    by one bit while $R_1$ remains unchanged.
  \end{itemize}

  This stage terminates achieving $R_1=C_1(S),
  R_2=\min\{C_2(S),C(S)-C_1(S)\}.$ Because the GNS-cut is not
  symmetric in indices $1$ and $2,$ we also have to show achievability
  of the rate pair $R_1=\min\{C_1(S),C(S)-C_2(S)\}, R_2=C_2(S).$ This
  can be shown similarly.

  \textbf{Case III}: $S$ is a minimal GNS-cut such that
  $\mathcal{G}\setminus S$ has no paths from $s_1$ to $t_1,$ $s_2$ to
  $t_2,$ or $s_1$ to $t_2$ but it has paths from $s_2$ to $t_1.$ This
  case is identical to Case II. \qed

\end{proof}

\subsection{Proof of Thm.~\ref{thm:two-unicast-hard}}
\label{subsec:two-unicast-hard}
\begin{proof}


  Let us split the number $\sum_{i=1}^k R_i = \sum_{j=1}^m r_j$ into a
  `coarsest common partition' formed by $c_1,c_2,\ldots, c_l$ as shown
  in Fig.~\ref{fig:split}. Set $c_0=0.$ Recursively define $c_h$ as
  the minimum of
  \begin{align}
    \min_{s: \sum_{i^\prime=1}^s R_{i^\prime} > \sum_{u=1}^{h-1}c_u}
    \sum_{i^\prime=1}^s R_{i^\prime}
    -\sum_{u=1}^{h-1}c_u\label{eq:min_i}
  \end{align}
  and
  \begin{align}
    \min_{t: \sum_{j^\prime=1}^t r_{j^\prime} > \sum_{u=1}^{h-1}c_u}
    \sum_{j^\prime=1}^t r_{j^\prime}
    -\sum_{u=1}^{h-1}c_u.\label{eq:min_j}
  \end{align}
  Note that by definition, $c_h>0$ even if some of the $R_i$ and/or
  $r_j$ were equal to $0.$ Define $l$ by
  $\sum_{u=1}^l c_u = \sum_{i=1}^k R_i.$ $l$ satisfies
  $\max\{k,m\}\leq l\leq k+m-1.$ We will alternately denote $c_h$ by
  $c_{(i,j)}$ where $i$ and $j$ are the $\arg\min$'s in
  \eqref{eq:min_i} and \eqref{eq:min_j} respectively.  We will use
  $\mathcal{I}$ to denote the indices $(i,j)$ that correspond to
  $c_{(i,j)}=c_h$ for some $h.$ In the rest of this proof, we will
  have $i, i_0$ denote an index belonging to $\{1,2,\ldots, k\},$
  $j, j_0$ denote an index belonging to $\{1,2,\ldots, m\}$ and
  $(i,j)$ or $(i_0,j_0)$ denote an index belonging to $\mathcal{I}.$
  Note that

  \begin{align}
    \sum_{j: (i,j)\in\mathcal{I}} c_{(i,j)} & = \sum_j c_{(i,j)} = R_i, \label{eq:sum-c-j}\\
    \sum_{i: (i,j)\in\mathcal{I}} c_{(i,j)} & = \sum_i c_{(i,j)} =
    r_j. \label{eq:sum-c-i}
  \end{align}
  
  \begin{figure}[h]
    \begin{center}
      \includegraphics[width = 4in, height=!]{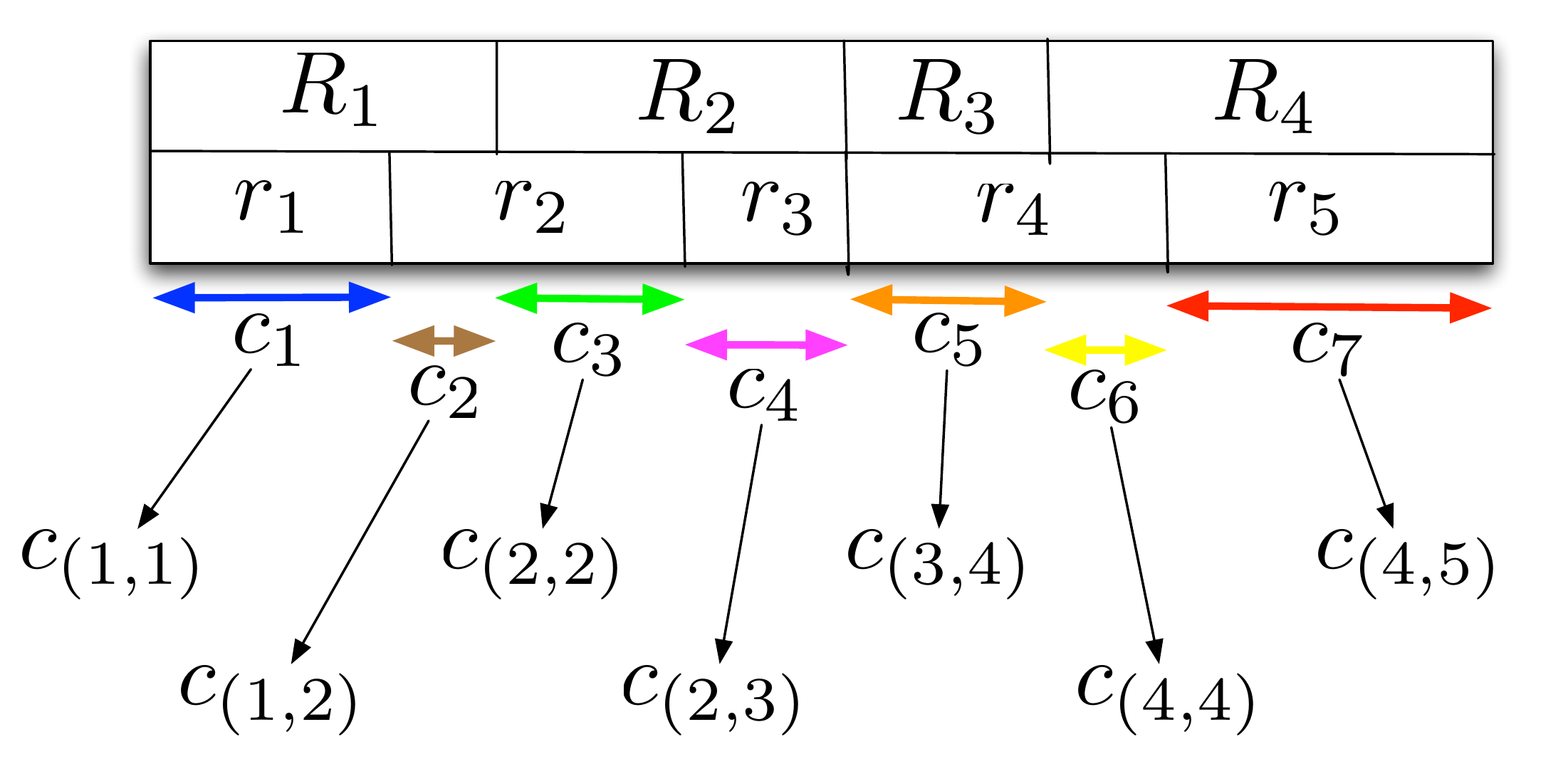}
      \caption{Splitting of the number $\sum_{i=1}^k R_i =
        \sum_{j=1}^m r_j$ to obtain $(c_h: h=1,2,\ldots,l)$}
      \label{fig:split}
    \end{center}
  \end{figure}
  
  Given a $(k+m)$-unicast network block $B$ with source-destination
  pairs $(s^\prime_h,t^\prime_h),$ $h=1,2,\ldots, k+m,$ we extend it
  into an $(m+1)$-unicast network $\mathcal{N}$ as follows:
  \begin{itemize}
  \item Create sources $s, s_1,s_2,\ldots, s_m$ and their
    corresponding destinations $t, t_1,t_2,\ldots, t_m.$ Create nodes
    $v_1,v_2,\ldots, v_m.$ Create nodes $x_{(i,j)}, y_{(i,j)},
    z_{(i,j)}, w_{(i,j)},w^1_{(i,j)}, w^2_{(i,j)}, w^3_{(i,j)} $ for
    each $(i,j)\in\mathcal{I}.$
  \item For $j=1,2,\ldots, m,$ create edges of capacity $R_{k+j}$ from
    $s_j$ to $v_j,$ $v_j$ to $s^\prime_{k+j},$ and $t^\prime_{k+j}$ to
    $t_j.$ (see Fig.~\ref{fig:general_construction}(a))
  \item For each $(i,j)\in\mathcal{I},$ create edges of capacity
    $c_{(i,j)}$ from $s$ to $x_{(i,j)}, x_{(i,j)}$ to $s^\prime_i,
    s_j$ to $y_{(i,j)}, t^\prime_i$ to $z_{(i,j)},$ as shown in
    Fig.~\ref{fig:general_construction}(a).  and the butterfly edges
    as shown in Fig.~\ref{fig:general_construction}(b).
  \end{itemize}

  \begin{figure}[htbp]
    \begin{center}
      \subfigure[]{\includegraphics[width = 4in,
        height=!]{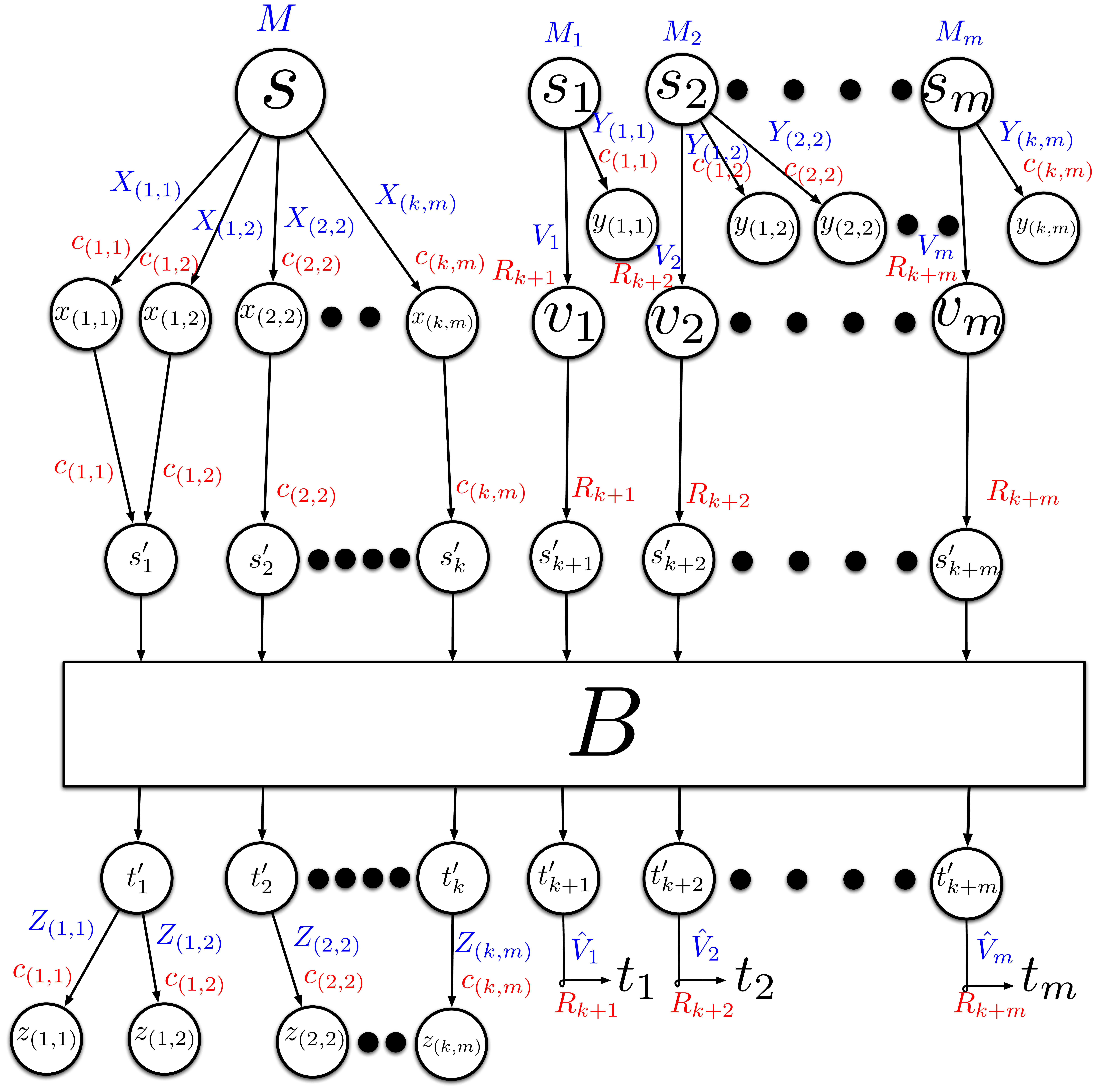}}\hspace{25pt}
      \subfigure[]{\includegraphics[width = 2in,
        height=!]{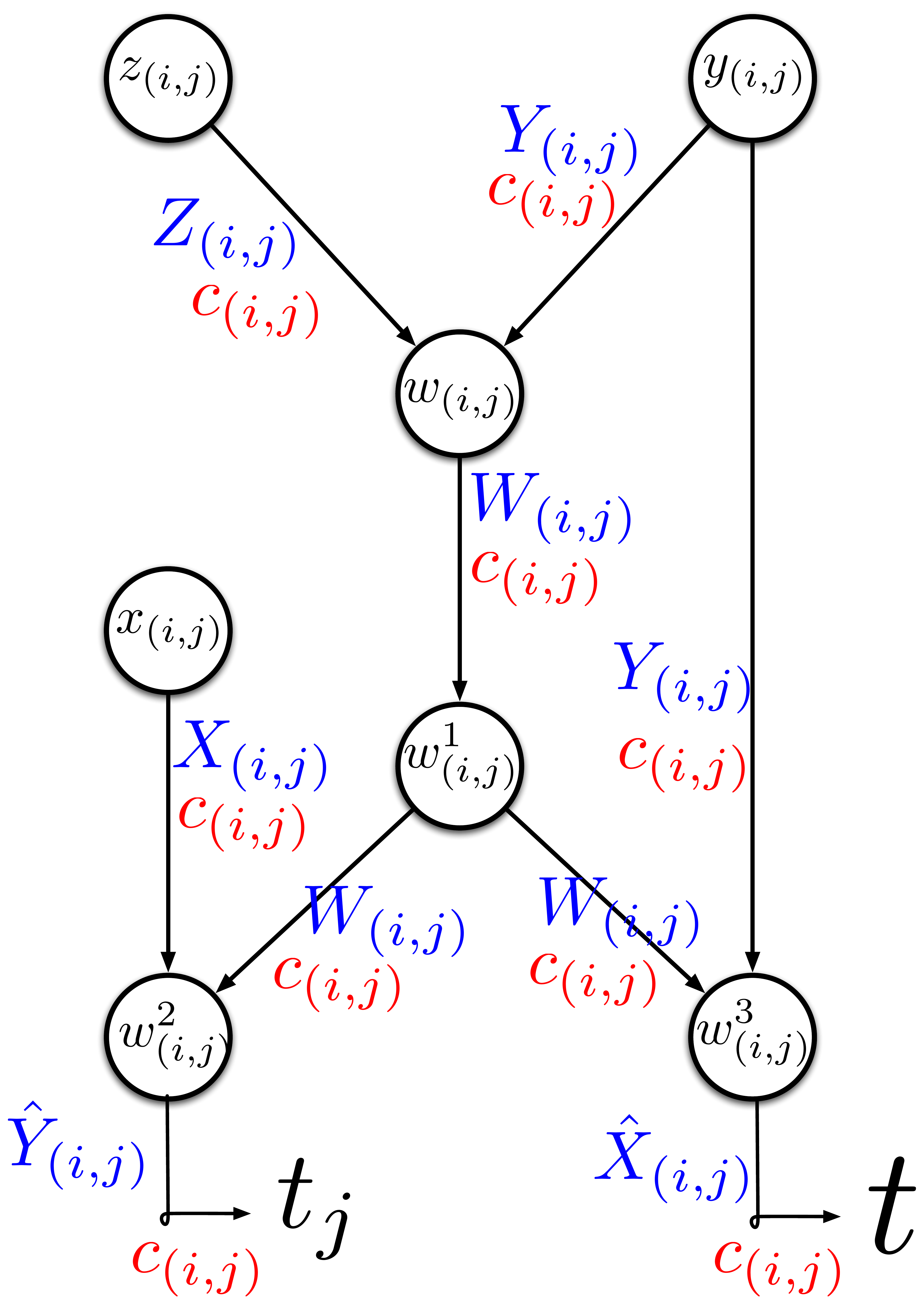}}
      \caption{(a) shows the $(m+1)$-unicast network $\mathcal{N}$
        constructed around a given $(k+m)$-unicast network block $B.$
        The label in red is the edge capacity and the label in blue is
        the random variable that flows through that edge. (b) shows
        the butterfly network component in the extended
        $(m+1)$-unicast network $\mathcal{N}$ for each
        $(i,j)\in\mathcal{I}.$ Each edge in this component has
        capacity $c_{(i,j)}.$ The label in red is the edge capacity
        and the label in blue is the random variable that flows
        through that edge.}
      \label{fig:general_construction}
    \end{center}
  \end{figure}

  We will prove that $(R_1,R_2,\ldots, R_{k+m})$ is achievable (or
  achievable by vector linear coding) in the $(k+m)$-unicast network
  $B$ if and only if $(\sum_{i=1}^k R_i, R_{k+1}, \ldots, R_{k+m})$ is
  achievable (respectively achievable by vector linear coding) in the
  $(m+1)$-unicast extended network $\mathcal{N}.$

  This will establish that any algorithm that can determine
  achievability of $(\sum_{i=1}^k R_i, R_{k+1}, \ldots, R_{k+m})$ may
  be applied to our extended network to determine achievability of
  $(R_1,R_2,\ldots, R_{k+m})$ in the network block $B.$  

  Suppose $(R_1,R_2,\ldots, R_{k+m})$ is achievable in the
  $(k+m)$-unicast network block $B.$ Then, we can come up with a
  `butterfly' coding scheme which proves the achievability of the rate
  tuple $(\sum_{i=1}^k R_i, R_{k+1}+r_1,R_{k+2}+r_2, \ldots,
  R_{k+m}+r_m)$ in the $(m+1)$-unicast network $\mathcal{N}$. This can
  be done simply by making $X_{(i,j)} = Z_{(i,j)}$ and performing
  butterfly coding over each butterfly network component. This means
  we set $W_{i,j} = Z_{i,j}+Y_{i,j}, \hat{Y}_{i,j} = X_{i,j} +
  W_{i,j}$ and $\hat{X}_{i,j} = W_{i,j} + Y_{i,j}.$ The addition here
  is done over an Abelian group $Z_l$ (summation modulo integer $l$)
  where $l$ is the size of the common finite alphabet that $X_{i,j},
  Y_{i,j}, W_{i,j}$ take values in, with the alphabet mapped to the
  Abelian group according to some arbitrary bijection.

  Now suppose that the rate tuple $(\sum_{i=1}^k R_i,
  R_{k+1}+r_1,R_{k+2}+r_2, \ldots, R_{k+m}+r_m)$ is achievable in the
  network $\mathcal{N}.$ We want to show that $(R_1,R_2,\ldots,
  R_{k+m})$ is achievable in the $(k+m)$-unicast network block $B.$ We
  will first present a plausibility argument for this.
  
  \emph{By tightness of the incoming rate constraint at $t_j$ and
    looking at all symbols that enter block $B$, we will inevitably
    require $\hat{V}_j$ to have all the information about $V_j.$
    Similarly, we will require $\hat{Y}_{(i,j)}, \hat{X}_{(i,j)}$ to
    have all the information about $Y_{(i,j)}, X_{(i,j)}$ respectively
    which will further necessitate that butterfly coding be done over
    the butterfly component and this will be possible only if
    $Z_{(i,j)}$ has all the information about $X_{(i,j)}.$}

  Now, if the capacity of any outgoing edge of a vertex is at least as
  large as the sum of the capacities of all incoming edges at that
  vertex, we will assume, without loss of generality, that the vertex
  sends all of its received data on that outgoing edge. We first state
  and prove a straightforward lemma that will be useful.

  \begin{lemma}\label{lem:simple}
    Suppose $A,B,C,D$ are random variables with $B,C,D$ mutually
    independent and satisfying $H(A|B,D) = 0.$ Then,
    \begin{itemize}
    \item[a)] $H(A|B,C) = 0 \implies H(A|B) = 0.$
    \item[b)] $H(B|A,C) = 0\implies H(B|A) = 0.$
    \end{itemize}
  \end{lemma}

  \begin{proof}
    $H(B,C,D) + H(A|B,C,D) = H(B,D) + H(A|B,D) + H(C|A,B,D).$ Mutual
    independence of $B,C,D$ gives $H(B,C,D) = H(B,D) + H(C).$ Further,
    $0\leq H(A|B,C,D) \leq H(A|B,D) = 0.$ So, we get $H(C) =
    H(C|A,B,D),$ i.e. $I(C;A,B,D) = 0.$ By the chain rule, this
    implies both $I(C;A|B) = 0$ and $I(C;B|A) = 0,$ i.e. $H(A|B) =
    H(A|B,C)$ and $H(B|A) = H(B|A,C).$ This completes the proof.
  \end{proof}

  Define random variables as shown in
  Fig.~\ref{fig:general_construction}, i.e. let $M$ denote the input
  message at source $s$ and for each $j=1,2,\ldots, m,$ let $M_j$
  denote the input message at source $s_j.$ Furthermore, let the
  random variables $V_j, \hat{V}_j$ for $j=1,2,\ldots, m,$ and
  $X_{(i,j)},$ $Y_{(i,j)},$ $Z_{(i,j)},$ $\hat{X}_{(i,j)},$
  $\hat{Y}_{(i,j)}$ for $(i,j)\in\mathcal{I}$ be as shown in
  Fig.~\ref{fig:general_construction}(a) and
  Fig.~\ref{fig:general_construction}(b).
  We will measure entropy with logarithms to the base
  $|\mathcal{A}|^N,$ where $\mathcal{A}$ is the alphabet and $N$ is
  the block length. $M, M_1,M_2,\ldots, M_m$ are mutually
  independent. $H(M) = \sum_{i=1}^k R_i,$ and for each $j,$ $H(M_j) =
  R_{k+j}+r_j.$ Let us write $A\leftrightarrow B$ if $H(A|B) = H(B|A)
  = 0.$ By observing the tight outgoing rate constraints and encoding
  at sources $s,s_1,s_2,\ldots, s_k,$ and the tight incoming rate
  constraints and decodability at destinations $t, t_1, t_2,\ldots,
  t_k,$ we can easily conclude \vspace{-0.1in}
  \begin{align}
    & M\leftrightarrow \cup_{(i,j)} \{X_{(i,j)}\} \leftrightarrow
    \cup_{(i,j)} \{\hat{X}_{(i,j)}\},\label{eq:M}\\
    & \forall j:M_j\leftrightarrow \cup_i \{Y_{(i,j)}\}\cup \{V_j\}
    \leftrightarrow \cup_i
    \{\hat{Y}_{(i,j)}\}\cup \{\hat{V}_j\},\label{eq:M_j}\\
    & \forall (i,j):\ \ \ H(X_{(i,j)}) =
    H(\hat{X}_{(i,j)}) = c_{(i,j)},\label{eq:X}\\
    & \forall (i,j):\ \ \ H(Y_{(i,j)}) = H(\hat{Y}_{(i,j)}) =
    c_{(i,j)},\label{eq:Y}\\
    &\forall j:\ \ \ H(V_j) = H(\hat{V}_j) = R_{k+j}.\label{eq:V}
  \end{align}
  \vspace{-0.1in} The random variables in the collection
  \begin{align}
    \cup_{(i,j)}\left(\{X_{(i,j)}\}\cup \{Y_{(i,j)}\}\right)\cup
    \left(\cup_j \{V_j\}\right)
    \label{eq:independence}
  \end{align}
  are mutually independent. In particular, \eqref{eq:X}, \eqref{eq:V}
  \eqref{eq:independence} imply that the messages received by the
  sources of the network block $s^\prime_h$ for $h=1,2,\ldots, k+m,$
  are mutually independent and the symbol received by $s^\prime_h$ has
  entropy $R_h.$

  Now, fix any $(i,j)\in\mathcal{I}.$
  \begin{align}
    & H(W_{(i,j)}|X_{(i,j)}) - H(W_{(i,j)}|M_j,X_{(i,j)}) \nonumber \\
    = & I(W_{(i,j)};M_j|X_{(i,j)}) \label{eq:-1}\\
    = & I(X_{(i,j)},W_{(i,j)};M_j) -
    I(X_{(i,j)};M_j) \\
    \stackrel{(a)}{=}&   I(X_{(i,j)},W_{(i,j)};M_j) - 0 \\
    \stackrel{(b)}{\geq} & I(\hat{Y}_{(i,j)};M_j) \\
    \stackrel{(c)}{=} & c_{(i,j)}, \label{eq:0}
  \end{align}
  where $(a)$ holds because $X_{(i,j)}$ is a function of $M$ and $M$
  is independent of $M_j,$ $(b)$ holds because $\hat{Y}_{(i,j)}$ is a
  function of $(X_{(i,j)}, W_{(i,j)}),$ and $(c)$ follows from
  \eqref{eq:M_j}, \eqref{eq:Y}, \eqref{eq:V}. Combining the inequality
  chain \eqref{eq:-1}-\eqref{eq:0} with the edge capacity constraint
  $H(W_{(i,j)}|X_{(i,j)})\leq H(W_{(i,j)}) \leq c_{(i,j)},$ we obtain
  \begin{align}
    H(W_{(i,j)}) & = c_{(i,j)}, \label{eq:1} \\
    H(W_{(i,j)}|X_{(i,j)}) & = c_{(i,j)}, \label{eq:2}\\
    H(W_{(i,j)}|X_{(i,j)},M_j) & = 0. \label{eq:3}
  \end{align}

  Similarly,
  \begin{align}
    & H(W_{(i,j)}|Y_{(i,j)}) - H(W_{(i,j)}|M, Y_{(i,j)}) \nonumber \\
    = &  I(W_{(i,j)};M|Y_{(i,j)}) \label{eq:4-}\\
    = & I(Y_{(i,j)}, W_{(i,j)};M) - I(Y_{(i,j)};M)  \\
    \stackrel{(d)}{=} &  I(Y_{(i,j)}, W_{(i,j)};M) - 0 \\
    \stackrel{(e)}{\geq} &  I(\hat{X}_{(i,j)};M) \\
    \stackrel{(f)}{=} & c_{(i,j)}, \label{eq:4}
  \end{align}
  where $(d)$ holds because $Y_{(i,j)}$ is a function of $M_j$ and
  $M_j$ is independent of $M,$ $(e)$ holds because $\hat{X}_{(i,j)}$
  is a function of $(Y_{(i,j)}, W_{(i,j)}),$ and $(f)$ follows from
  \eqref{eq:M}, \eqref{eq:X}. Combining the inequality chain
  \eqref{eq:4-}-\eqref{eq:4} with the edge capacity constraint
  $H(W_{(i,j)}|Y_{(i,j)})\leq H(W_{(i,j)}) \leq c_{(i,j)},$ we obtain
  \begin{align}
    H(W_{(i,j)}|Y_{(i,j)}) & = c_{(i,j)}, \label{eq:6}\\
    H(W_{(i,j)}|Y_{(i,j)},M) & = 0. \label{eq:7}
  \end{align}

  From \eqref{eq:M_j}, we may rewrite \eqref{eq:3} as
  \begin{align}
    H(W_{(i,j)}|X_{(i,j)},V_j,\cup_{i_0} \{Y_{(i_0,j)}\}) =
    0. \label{eq:8}
  \end{align}
  From \eqref{eq:M}, we may rewrite \eqref{eq:7} as
  \begin{align}
    H(W_{(i,j)}|Y_{(i,j)},\cup_{i_0,j_0} \{X_{(i_0,j_0)}\}) =
    0. \label{eq:9}
  \end{align}

  Using Lemma~\ref{lem:simple}.a) with $A=W_{(i,j)},
  B=\{X_{(i,j)},Y_{(i,j)}\}, C=\{V_j\}\cup_{i_0}
  \{Y_{(i_0,j)}\}\setminus\{Y_{(i,j)}\}, D=\cup_{i_0,j_0}
  \{X_{(i_0,j_0)}\}\setminus\{X_{(i,j)}\},$ and using
  \eqref{eq:independence}, \eqref{eq:8}, \eqref{eq:9}, we obtain
  \begin{align}
    H(W_{(i,j)}|Y_{(i,j)},X_{(i,j)}) = 0. \label{eq:10}
  \end{align}

  Now, by the chain rule for entropy,
  \begin{align}
    & H(Y_{(i,j)}) + H(W_{(i,j)}|Y_{(i,j)}) +
    H(X_{(i,j)}|W_{(i,j)},Y_{(i,j)})  \nonumber \\
    & = H(X_{(i,j)},Y_{(i,j)}) +
    H(W_{(i,j)}|X_{(i,j)},Y_{(i,j)}) \label{eq:11}
  \end{align}
  Using \eqref{eq:independence}, \eqref{eq:6}, \eqref{eq:X},
  \eqref{eq:Y}, \eqref{eq:10} in \eqref{eq:11}, we get
  \begin{align}
    H(X_{(i,j)}|W_{(i,j)},Y_{(i,j)}) = 0. \label{eq:12}
  \end{align}

  From the encoding constraint at node $w_{(i,j)},$ we have
  \begin{align}
H(W_{(i,j)}|Y_{(i,j)},Z_{(i,j)}) = 0. \label{eq:13}
\end{align}
Putting \eqref{eq:12} and \eqref{eq:13} together, we get 

\begin{align}
H(X_{(i,j)}|Y_{(i,j)},Z_{(i,j)}) = 0. \label{eq:14}
\end{align}
From the encoding constraint for network block $B,$ we have that
\begin{align}
H(Z_{(i,j)}|\cup_{j_0} \{V_{j_0}\} \cup_{(i_0,j_0)} \{X_{i_0,j_0}\}) = 0. \label{eq:15}
\end{align}

Using Lemma~\ref{lem:simple}.b) with $A=Z_{(i,j)}, B=X_{(i,j)},
C=Y_{(i,j)}, D=\cup_{j_0} \{V_{j_0}\} \cup_{(i_0,j_0)}
\{X_{i_0,j_0}\}\setminus\{X_{(i,j)}\}$ and using
\eqref{eq:independence}, \eqref{eq:14}, \eqref{eq:15}, we obtain
\begin{align}
H(X_{(i,j)}|Z_{(i,j)}) = 0. \label{eq:16}
\end{align}

From \eqref{eq:independence}, \eqref{eq:X}, we have that for any
$i=1,2,\ldots, k,$ $H(\cup_j \{X_{(i,j)}\}) = R_i$ and \eqref{eq:16}
implies $H(\cup_j \{X_{(i,j)}\}|\cup_j \{Z_{(i,j)}\}) = 0.$ This shows
that in the block $B,$ the destination $t^\prime_i$ can decode source
$s^\prime_i$'s message for $i=1,2,\ldots, k.$

Fix any $j=1,2,\ldots, m.$ From \eqref{eq:M_j}, we have
\begin{align}
  H(\hat{V}_j|M_j) & = 0 \label{eq:20} 
\end{align}
By using \eqref{eq:M_j}, we can rewrite \eqref{eq:20} as 
\begin{align}
  \label{eq:21}
  H(\hat{V}_j|V_j, \cup_{i_0} \{Y_{(i_0,j)}\}) & = 0.
\end{align}
By the encoding constraint provided by the network block $B,$ we get
\begin{align}
  \label{eq:22}
  H\left(\hat{V}_j|\cup_{i_0,j_0}
    X_{(i_0,j_0)},\cup_{j_0}\{V_{j_0}\}\right) = 0.
\end{align}
Using Lemma~\ref{lem:simple}.a) with $A=\hat{V}_j, B=V_j, C=\cup_{i_0}
\{Y_{(i_0,j)}\}, D=(\cup_{i_0,j_0} X_{(i_0,j_0)})\cup
(\cup_{j_0}\{V_{j_0}\} \setminus \{V_j\}),$ and using
\eqref{eq:independence}, \eqref{eq:21} and \eqref{eq:22}, we obtain
\begin{align}
  \label{eq:23}
  H(\hat{V}_j|V_j) = 0.
\end{align}
By the chain rule for entropy,
\begin{align}
\label{eq:24}
H(V_j|\hat{V}_j) + H(\hat{V}_j) & = H(\hat{V}_j|V_j) +
H(V_j).
\end{align}
Using \eqref{eq:V} and \eqref{eq:23} in \eqref{eq:24},
we get 
\begin{align}
  \label{eq:25}
  H(V_j|\hat{V}_j) =0.
\end{align}

This implies that the destination $t_{k+j}^\prime$ can decode
$s_{k+j}^\prime$'s message for $j=1,2,\ldots, m$ within block $B.$

The case of the rate tuples assumed to be achievable by a vector
linear coding scheme is identical. This completes the proof.

\end{proof}

\begin{remark}\label{rem:linear-codes}
  The proof above used the notion of zero-error exactly achievable
  rates. It can be shown to go through for the notion of zero-error
  asymptotically achievable linear coding rates and vanishing error
  linear coding rates as follows.
  
  First, a very simple argument yields that if a linear code makes an error
  with positive probability, then the error probability is at least
  $\frac 12.$ Let $g$ be the global decoding function for the network,
  i.e. the function that maps all source messages to all decoded
  messages at the respective destinations (as in
  Definition~\ref{def:achievable}). If the code used is linear, then
  $g$ is a linear map over some finite field and so, is
  $g-\mathsf{id}$ where $\mathsf{id}$ is the identity mapping. Let the
  null space of this map be $S = \{v: g(v)-v=0\}.$ If $g(v_0)\neq v_0$
  for some $v_0,$ then
  \begin{align}
    g(v)\neq v,\ \ \forall v\in v_0+S:=\{v_0+v_1: v_1\in S\}.
  \end{align}
  As $S$ and $v_0+S$ are disjoint, $\mathsf{Pr}(v\in S) +
  \mathsf{Pr}(v\in v_0+S)\leq 1.$ But $\mathsf{Pr}(v\in S) =
  \mathsf{Pr}(v\in v_0+S).$ So, $\mathsf{Pr}(v\in S)\leq \frac 12.$
  The probability of error is $\mathsf{Pr}(v\not\in S)\geq \frac 12.$
  This means that the zero-error asymptotically achievable linear
  coding capacity is identical to the zero-error linear coding Shannon
  capacity for any multiple unicast network.

  To show that the reduction works for zero-error asymptotically
  achievable rates, note that if for any $\epsilon>0,$ the rate tuple
  $(\sum_{i=1}^k R_i-\epsilon,
  R_{k+1}+r_1-\epsilon,R_{k+2}+r_2-\epsilon,\ldots, R_{k+m} +
  r_m-\epsilon)$ is zero-error asymptotically achievable by linear
  codes, then every single equality and inequality stated in the proof
  continues to hold upto a correction term $\delta(\epsilon)$ where
  $\delta(\epsilon)\to 0$ as $\epsilon\to 0.$ From linearity, we can
  find a suitable choice of subspaces at $s_1^\prime, s_2^\prime,
  \ldots, s^\prime_{k+m}$ so that each source $s_i^\prime$ is
  transmitting an independent message at rate
  $R_i-\delta^\prime(\epsilon),$ where $\delta^\prime(\epsilon)\to 0$
  as $\epsilon\to 0.$

  This completes the justification of the summary in
  Table~\ref{table:reduction}. In particular, the reduction works in
  full generality for \emph{linear codes}.  There are technical
  difficulties with showing that the reduction goes through for
  zero-error asymptotically achievable rates or vanishing error
  achievable rates when \emph{general codes} are used. It is not known
  whether the reduction will work in these cases.
\end{remark}

\subsection{Proof of Theorem~\ref{thm:insufficiency}}
\label{subsec:insufficiency}

\begin{proof}
  Consider the network in Fig.~3 of \cite{Zeger_insufficiency}, for
  which linear codes were shown to be insufficient to achieve
  capacity. This network can be converted into a 10-unicast network
  using the construction in \cite{Zeger_nonreversibility}. By applying
  our construction, we find a two-unicast network in which the rate
  point $(9,10)$ is achievable by non-linear codes but not by linear
  codes.
\end{proof}

\subsection{Proof of Theorem~\ref{thm:non-Shannon}}
\label{subsec:non-Shannon}

\begin{proof}
  Consider the 6-unicast V\'{a}mos network in Fig.~13 of
  \cite{Zeger_matroid}. This network is shown in \cite{Zeger_matroid}
  to be matroidal (as defined by them in Definition~V.1) and hence,
  the best bound that can be obtained using Shannon inequalities on
  the maximum uniform rate achievable is 1. However, a unit rate per
  unicast is shown to not be achievable for this network. By applying
  our construction treating the V\'{a}mos network as a block $B,$ we
  get a two-unicast network with desired rate pair $(5,6).$ This can
  be naturally viewed as an 11-unicast network $\mathcal{N}$ with unit
  rate requirement for each unicast session.

  Consider the auxiliary 11-unicast network $\mathcal{N}^\prime$
  obtained by removing the block $B$ from $\mathcal{N}$ and fusing
  source-destination nodes of $B$, i.e. fusing $s_i^\prime$ with
  $t_i^\prime$ for $i=1,2,\ldots, 6.$ It is easy to verify that the
  butterfly coding scheme proposed in the proof of
  Theorem~\ref{thm:two-unicast-hard} can provide a linear coding scheme that
  achieves unit rate for each of the 11 source-destination pairs of
  $\mathcal{N}^\prime.$ This linear coding scheme naturally makes
  $\mathcal{N}^\prime$ a matroidal network (as elaborated in Theorem
  V.4 of \cite{Zeger_matroid}).

  The restriction of the matroids corresponding to the network in
  block $B$ and the network $\mathcal{N}^\prime$ to the intersection
  of their ground sets yields the full rank matroid on six
  elements. This is obviously a modular matroid, as defined in
  Sec. IV-A of \cite{Zeger_matroid}. Using Lemma~IV.7 from
  \cite{Zeger_matroid} (which guarantees the existence of a proper
  amalgam of matroids as long as their restrictions to the
  intersection of the ground sets is modular), we obtain that the
  11-unicast network $\mathcal{N}$ is matroidal too.

  Since the 11-unicast network $\mathcal{N}$ is matroidal, Shannon
  inequalities cannot rule out achievability of unit rate for each of
  the 11-unicast sessions or alternately, achievability of $(5,6)$ for
  the two-unicast network. However, if $(5,6)$ was achievable in the
  two-unicast network, then the construction would imply achievability
  of unit rate for each session of the 6-unicast V\'{a}mos network,
  which is proved to be impossible in \cite{Zeger_matroid}. Hence,
  $(5,6)$ is not achievable in the two-unicast network constructed
  from the V\'{a}mos network.
\end{proof}

\end{document}